\documentclass[11pt]{article}

\usepackage[T1]{fontenc}
\usepackage{lmodern}
\usepackage{microtype}
\usepackage[letterpaper,margin=1in]{geometry}

\usepackage{amsmath,amssymb,amsthm,mathtools}
\usepackage{aliascnt}
\usepackage{array}
\usepackage{booktabs}
\usepackage{enumitem}
\usepackage{needspace}
\usepackage{algorithm,algpseudocode}
\algrenewcommand\algorithmicrequire{\textbf{Input:}}
\algrenewcommand\algorithmicensure{\textbf{Output:}}
\usepackage{tikz}
\usepackage{xcolor}
\usetikzlibrary{arrows.meta,backgrounds,fit,positioning}

\definecolor{linkblue}{RGB}{25,75,130}
\usepackage[
  colorlinks=true,
  linkcolor=linkblue,
  citecolor=linkblue,
  urlcolor=linkblue
]{hyperref}
\usepackage[nameinlink,noabbrev,capitalise]{cleveref}

\allowdisplaybreaks
\setlist{itemsep=0.25em,topsep=0.5em}
\newtheorem{theorem}{Theorem}[section]
\newaliascnt{lemma}{theorem}
\newtheorem{lemma}[lemma]{Lemma}
\aliascntresetthe{lemma}
\newaliascnt{proposition}{theorem}
\newtheorem{proposition}[proposition]{Proposition}
\aliascntresetthe{proposition}
\newaliascnt{corollary}{theorem}
\newtheorem{corollary}[corollary]{Corollary}
\aliascntresetthe{corollary}
\theoremstyle{definition}
\newaliascnt{definition}{theorem}
\newtheorem{definition}[definition]{Definition}
\aliascntresetthe{definition}
\theoremstyle{remark}
\newaliascnt{remark}{theorem}
\newtheorem{remark}[remark]{Remark}
\aliascntresetthe{remark}

\crefname{theorem}{Theorem}{Theorems}
\crefname{lemma}{Lemma}{Lemmas}
\crefname{proposition}{Proposition}{Propositions}
\crefname{corollary}{Corollary}{Corollaries}
\crefname{definition}{Definition}{Definitions}
\crefname{remark}{Remark}{Remarks}

\numberwithin{equation}{section}

\AddToHook{cmd/thebibliography/after}{\setlength{\itemsep}{0pt}}

\newcommand{\bits}{\{0,1\}}
\newcommand{\I}{\mathrm{I}}
\newcommand{\J}{\mathrm{J}}
\newcommand{\ones}{\mathbf{1}}
\newcommand{\zeros}{\mathbf{0}}
\newcommand{\transpose}{\mathsf{T}}
\newcommand{\supp}{\operatorname{supp}}
\newcommand{\im}{\operatorname{Im}}

\newcommand{\BMPV}{\operatorname{\mathsf{BMPV}}}
\newcommand{\OV}{\mathsf{OV}}
\newcommand{\promOV}{\mathsf{PromOV}}

\newcommand{\BMPRS}{\operatorname{\mathsf{BMPRS}}}
\newcommand{\DIAM}{\operatorname{\mathsf{DIAM}}}
\newcommand{\RAD}{\operatorname{\mathsf{RAD}}}
\newcommand{\Triangle}{\operatorname{\mathsf{Triangle}}}
\newcommand{\ONTO}{\operatorname{\mathsf{ONTO}}}
\newcommand{\AND}{\operatorname{\mathsf{AND}}}
\newcommand{\OR}{\operatorname{\mathsf{OR}}}
\newcommand{\Adv}{\operatorname{Adv}^{\!\pm}}

\newcommand{\booland}{\mathbin{\wedge}}
\newcommand{\boolor}{\mathbin{\vee}}
\newcommand{\boolprod}{\mathbin{\bullet}}

\newcommand{\ind}[1]{\mathbf{1}\!\left[#1\right]}

\newcommand{\ket}[1]{\left\lvert #1\right\rangle}
\newcommand{\ceil}[1]{\left\lceil#1\right\rceil}
\newcommand{\floor}[1]{\left\lfloor#1\right\rfloor}
\newcommand{\from}{\colon}
\newcommand{\polylog}{\operatorname{polylog}}

\newcommand{\Ot}{\widetilde O}
\newcommand{\Left}[1]{\mathsf{Left}_{#1}}
\newcommand{\Right}[1]{\mathsf{Right}_{#1}}

\newcommand{\Remaining}[1]{\mathsf{Remaining}_{#1}}

\newcommand{\Heavy}[1]{\mathsf{Heavy}_{#1}}
\newcommand{\Size}[1]{\mathsf{size}_{#1}}
\newcommand{\Work}[1]{\mathsf{work}_{#1}}
\newcommand{\Budget}[1]{\mathsf{budget}_{#1}}
\newcommand{\First}[2]{\operatorname{first}(#1,#2)}
\newcommand{\Record}[1]{\operatorname{record}(#1)}
\newcommand{\Pivots}{\mathsf{Pivots}}
\providecommand{\Covered}[1]{\mathsf{Covered}_{#1}}
\providecommand{\DegSamp}{\mathsf{DegSamp}}

\title{Improved Quantum Query Bounds for\\ Boolean Matrix Product Verification}

\author{Amin Shiraz Gilani\thanks{University of Maryland.
\href{mailto:asgilani@umd.edu}{\texttt{asgilani@umd.edu}}}
\qquad
Fran\c{c}ois Le Gall\thanks{Graduate School of Mathematics, Nagoya University.
\href{mailto:legall@math.nagoya-u.ac.jp}{\texttt{legall@math.nagoya-u.ac.jp}}}
\qquad
Xingyu Zhou\thanks{The University of British Columbia.
\href{mailto:zxingyu@cs.ubc.ca}{\texttt{zxingyu@cs.ubc.ca}}}}
\date{}

\hypersetup{
  pdftitle={Improved Quantum Query Bounds for Boolean Matrix Product
    Verification},
  pdfauthor={Amin Shiraz Gilani, Fran\c{c}ois Le Gall, and Xingyu Zhou},
}

\begin{document}

\maketitle
\begin{abstract}
We prove the first non-trivial upper bound for the quantum query complexity of Boolean Matrix Product Verification ($\mathsf{BMPV}$), answering a longstanding open question in quantum query complexity. For $n\times n$ matrices, our upper bound is $\widetilde O(n^{17/12})$, improving on the standard $O(n^{3/2})$ bound obtained using Grover search by Buhrman and Špalek [SODA 2006]. We complement this result by showing an $\Omega(n^{5/4})$ lower bound, which improves over the previous best known lower bound of $\widetilde\Omega(n^{19/18})$ by  Childs, Kimmel, and Kothari [ESA 2012].

Our approach centers on a connection with Orthogonal Vectors ($\mathsf{OV}$), which asks whether an indexed list of $n$ Boolean vectors of dimension $n$ contains two vectors with disjoint supports. 
In particular, we prove equivalences between $\mathsf{OV}$ and $\mathsf{BMPV}$ and establish the above bounds for $\mathsf{OV}$.
We also prove a tight $\widetilde \Theta(n^{3/2})$ bound for a variant of $\mathsf{BMPV}$ that asks whether the product contains a given row vector. Together, these results imply a polynomial separation between the quantum query complexities of deciding whether a graph has radius at most two and whether it has diameter at most two. 
\end{abstract}

\clearpage
\begingroup
\setcounter{tocdepth}{2}
\tableofcontents
\endgroup
\clearpage

\section{Introduction}
\label{sec:introduction}

\subsection{Background and statement of the main results}
\paragraph{Boolean Matrix Product Verification.}
Boolean matrix product verification (\(\BMPV_n\)) is one of the central problems in fine-grained complexity~\cite{VassilevskaWilliamsWilliams2018,Kunnemann2018Freivalds}. Given as input three matrices \(A,B,C\in\bits^{n\times n}\), this problem asks whether \(A\boolprod B=C\), where $\boolprod$ denotes the Boolean matrix product, i.e., the matrix product over the Boolean semiring in which addition is OR and multiplication is AND. While over rings Freivalds's randomized linear fingerprint verifies a matrix product classically in quadratic time~\cite{Freivalds1979}, the same argument does not extend to the Boolean semiring. The best known classical algorithm for Boolean matrix product verification works by reducing it to matrix multiplication, leading to time complexity $O(n^{\omega+o(1)})$, where $\omega<2.38$ denotes the exponent of matrix multiplication.

Even in the quantum query model,\footnote{Boolean matrix product verification is not interesting in the classical query model since its classical query complexity is trivially $\Theta(n^2)$.} in which the algorithm receives
oracle access to \(A,B,C\) and only queries to matrix entries are counted, the
complexity of verifying a Boolean matrix product remains unresolved.
Grover search immediately gives the upper bound
\(Q(\BMPV_n)=O(n^{3/2})\), as observed by Buhrman and \v{S}palek \cite{BuhrmanSpalek2006}, where~$Q$ denotes bounded-error quantum query complexity. 
Whether this upper bound is optimal has remained a longstanding open question in quantum query complexity (see, e.g., \cite{ChildsKimmelKothari2012,JefferyKothariMagniez2012,KothariNayak2016,LeGall2014Survey} and \cite[Open Problem 4.5]{Kothari2014}). The best known lower bound is $Q(\BMPV_n)=\widetilde\Omega(n^{19/18})=\widetilde\Omega(n^{1.0555\ldots})$ proved by Childs, Kimmel, and Kothari using the read-many formula method~\cite{ChildsKimmelKothari2012}, leaving a large gap.

The quantum query complexity of Boolean matrix product verification is particularly interesting because of its connection to triangle finding, a central problem that has driven research in quantum query complexity~\cite{MagniezSanthaSzegedy2007,LeeMagniezSantha2013,LeGall2014Triangle,CaretteLauriereMagniez2017}. Although triangle finding reduces to Boolean matrix product verification, successive improvements have led to quantum algorithms for triangle finding in \(n\)-vertex graphs using \(\widetilde O(n^{5/4})\) queries, well below the \(n^{3/2}\) bound obtained by Grover search. A significant difference between the two problems lies in their certificate complexities. Whereas triangle finding has constant-size \(1\)-certificates, the \(0\)- and \(1\)-certificate complexities of \(\BMPV_n\) are \(\Theta(n)\) and \(\Theta(n^2)\), respectively. This contrast suggests that Boolean matrix product verification may be harder than triangle finding and, prior to this work, could have suggested that its \(O(n^{3/2})\) upper bound was optimal. The previously known lower bounds also reflect this contrast: Boolean matrix product verification had a \(\widetilde\Omega(n^{19/18})\) lower bound, whereas the best known lower bound for triangle finding remains the trivial \(\Omega(n)\).




In this work, we show the first non-trivial upper bound for Boolean Matrix Product Verification:\footnote{In this paper, the notations \(\widetilde O(\cdot)\),  \(\widetilde \Omega(\cdot)\) and  \(\widetilde \Theta(\cdot)\) suppress $\polylog(n)$ factors.}
\begin{theorem}\label{th:main_ub}
  $Q(\BMPV_n)=\widetilde O(n^{17/12})$. (Note that $17/12=1.41666\ldots$.)
\end{theorem}

We complement this upper bound by the following new lower bound: 

\begin{theorem}\label{th:main_lb}
  $Q(\BMPV_n)=\Omega(n^{5/4})$.
\end{theorem}

Theorem \ref{th:main_lb} significantly improves the $\widetilde\Omega(n^{19/18})$ lower bound from \cite{ChildsKimmelKothari2012}. Our lower bound holds even when \(C\) is the all-ones matrix
and \(B=A^\transpose\). 

\paragraph{Diameter and radius.}
We also study two graph parameters: the diameter (the maximum distance between any pair of vertices) and the radius (the minimum, over all vertices, of the maximum distance from that vertex to any other vertex). Throughout, graphs are unweighted and undirected, and $n$ denotes the number of vertices.


Computing and approximating graph diameter and radius have long been studied in classical algorithms and fine-grained complexity~\cite{AingworthChekuriIndykMotwani1999,BaschKhannaMotwani1995,RodittyVassilevskaWilliams2013,ChechikEtAl2014,CairoGrossiRizzi2016,AbboudVassilevskaWilliamsWang2016,BackursEtAl2018}. The special case of deciding whether the diameter is at most two, which we denote $\DIAM_{\le2,n}$,  is closely connected to Boolean matrix product verification: this property holds precisely when $M\boolprod M$ is the all-ones matrix, where $M$ is the adjacency matrix with ones added on the diagonal. Deciding whether the radius is at most two, which we denote $\RAD_{\le2,n}$, also has a particularly simple matrix formulation: it asks whether $M\boolprod M$ contains an all-ones row. The best known classical algorithms for both problems have time complexity $O(n^{\omega+o(1)})$.

We investigate these decision problems in the quantum query model, with oracle access to the graph’s adjacency matrix.
We apply the techniques developed in this paper and show that the bounds of Theorems~\ref{th:main_ub} and~\ref{th:main_lb} also hold for deciding whether the diameter is at most two:
\begin{theorem}\label{th:main_diam}
  $Q(\DIAM_{\le2,n})=\widetilde O(n^{17/12})$ and 
  $Q(\DIAM_{\le2,n})=\Omega(n^{5/4})$.
  \label{eq:intro-bmpv-bounds}
\end{theorem}
We then give bounds, tight up to polylogarithmic factors, for the quantum query complexity of radius at most two:
\begin{theorem}\label{th:main_rad}
   $Q(\RAD_{\le2,n})=\widetilde \Theta (n^{3/2})$.
\end{theorem}
These results give a polynomial separation between the quantum query complexities of deciding whether the radius is at most two and whether the diameter is at most two.
Interestingly, such a separation is not known in the classical time complexity setting. Note that the main difference between the two problems is a change of quantifier. While $\DIAM_{\le2,n}$ corresponds to $\BMPV_n$ with the all-ones target matrix, $\RAD_{\le2,n}$ corresponds to the all-ones-target case of Boolean Matrix Product Row Search (\(\BMPRS_n\)), which asks whether \emph{some} row of \(A\boolprod B\) equals a given vector \(v\). For \(v=\ones_n\), this becomes
\begin{equation}
  \BMPRS_n(A,B,\ones_n)=1
  \quad\Longleftrightarrow\quad
  \exists i\in[n]\,\forall j\in[n]\,\exists t\in[n]:
  A_{it}=B_{tj}=1.
  \label{eq:intro-exists-forall-exists}
\end{equation}
Our results also show that $\BMPRS$ 
has quantum query complexity $\widetilde \Theta (n^{3/2})$, yielding the same polynomial separation from $\BMPV$.


\subsection{Proof overview}
\label{sec:proof-overview}

\Cref{fig:problem-landscape} summarizes our results.

\begin{figure}[!htbp]
  \centering
  \begin{tikzpicture}[
      problem/.style={
        draw,
        rounded corners=2pt,
        minimum height=8mm,
        text width=28mm,
        align=center,
        inner xsep=3pt,
        font=\small
      },
      family/.style={
        draw=black!30,
        fill=black!2,
        rounded corners=3pt,
        inner xsep=8pt,
        inner ysep=7pt
      },
      family label/.style={font=\footnotesize\bfseries},
      bounds/.style={font=\footnotesize,anchor=north},
      edge label/.style={
        fill=white,
        inner sep=1.5pt,
        align=center,
        font=\footnotesize
      },
      equivalence/.style={<->,>=Stealth,semithick},
      conceptual/.style={dashed,black!55}
    ]
    \node[problem] (ov) at (0,0) {\(\OV\)};
    \node[problem,left=28mm of ov] (bmpv) {\(\BMPV\)};
    \node[problem,right=28mm of ov] (diameter) {
      Diameter at Most Two
    };
    \node[bounds] (pairbounds) at ([yshift=-4mm]ov.south) {
      Lower bound: \(\Omega(n^{5/4})\)
      \qquad Upper bound: \(\widetilde O(n^{17/12})\)
    };

    \coordinate (existcenter) at ([yshift=-44mm]ov.center);
    \node[problem] (bpr) at (bmpv |- existcenter) {Boolean\\Matrix Product\\Row Search};
    \node[problem] (radius) at (diameter |- existcenter) {Radius at Most Two};
    \coordinate (existbase) at (ov |- bpr.south);
    \node[bounds] (existbounds) at ([yshift=-4mm]existbase) {
      Lower bound: \(\Omega(n^{3/2}/\sqrt{\log n})\)
      \qquad Upper bound: \(O(n^{3/2})\)
    };

    \begin{scope}[on background layer]
      \node[family,fit=(bmpv)(ov)(diameter)(pairbounds),
        label={[family label]north:Pairwise intersection problems}] {};
      \node[family,fit=(bpr)(radius)(existbounds),
        label={[family label]north:Existential family}] {};
    \end{scope}

    \draw[equivalence] (bmpv) -- (ov);
    \draw[equivalence] (ov) -- (diameter);
    \draw[equivalence] (bpr) -- node[edge label,above] {\(v=\ones_n\)} (radius);

    \coordinate (lefttop) at ([xshift=-6mm]bmpv.west);
    \coordinate (leftbottom) at ([xshift=-6mm]bpr.west);
    \draw[conceptual]
      (bmpv.west) -- (lefttop) --
      node[edge label,rotate=90] {every row versus some row}
      (leftbottom) -- (bpr.west);
    \coordinate (righttop) at ([xshift=6mm]diameter.east);
    \coordinate (rightbottom) at ([xshift=6mm]radius.east);
    \draw[conceptual]
      (diameter.east) -- (righttop) --
      node[edge label,rotate=-90] {every vertex versus some vertex}
      (rightbottom) -- (radius.east);
  \end{tikzpicture}
  \caption{Query bounds and relationships among the problems.  Solid arrows
  indicate equivalence in quantum query complexity up to constant-factor
  changes in problem dimension.
  The lower arrow restricts \(\BMPRS\) to the all-ones target.
  Dashed connectors change ``every''
  to ``some,'' fixing \(C=\J_n\) and \(v=\ones_n\) on the matrix side.
  Graphs use adjacency matrix queries.}
  \label{fig:problem-landscape}
\end{figure}

\subsubsection{Reductions and the lower bound}

\paragraph{Reductions to Orthogonal Vectors.}
Our central connection is with Orthogonal Vectors.  The problem
\(\OV_n\) asks whether an indexed list of \(n\) vectors in
\(\bits^n\) contains two vectors at distinct positions with disjoint
supports.  

The inequality 
\begin{equation}
  Q(\OV_n)\le Q(\BMPV_n)
  \label{eq:intro-bmpv-1}
\end{equation}
is straightforward. We show
\begin{equation}
  Q(\OV_n)\le Q(\BMPV_n)
  =O\!\left(n^{5/4}+Q(\OV_{2n+2})\right)
  \label{eq:intro-bmpv-2}
\end{equation}
by observing (see also \cite[Section~4.3]{Kothari2014}) that
\[
  A\boolprod B=C
  \quad\Longleftrightarrow\quad
  A\boolprod B\le C\ \text{ and }\ A\boolprod B\ge C,
\]
where \(X\le Y\) means that 
\(X_{ij}\le Y_{ij}\) for every \(i,j\), and \(X\ge Y\) is defined analogously.
A failure of the first inequality supplies a triple \((i,t,j)\) with
\(A_{it}=B_{tj}=1\) and \(C_{ij}=0\), and is equivalent to a triangle
in a standard tripartite graph, which can be tested in
\(O(n^{5/4})\) queries~\cite{CaretteLauriereMagniez2017}.
A failure of the second inequality can be encoded as an
\(\OV_{2n+2}\) instance.  

Our main technical contribution is showing that $Q(\OV_n)=\Omega(n^{5/4})$ and $Q(\OV_n)=\widetilde O(n^{17/12})$.
The lower bound 
comes from surjectivity, and the upper bound combines variable-time
search with a quantum walk for detecting an orthogonal pair.
Theorems \ref{th:main_ub} and \ref{th:main_lb} follow by combining these bounds with \cref{eq:intro-bmpv-1} and \cref{eq:intro-bmpv-2}.

We also show
\begin{equation}
  Q(\DIAM_{\le2,n})\le Q(\OV_n)
  \le Q(\DIAM_{\le2,2n}),
  \label{eq:intro-diameter}
\end{equation} 
which proves Theorem \ref{th:main_diam}.

Closed neighborhoods give the reduction from diameter to \(\OV\).
For the reverse diameter reduction, represent rows by independent
vertices and columns by a clique, with incidence edges supplied
by the input bits.  Two row vertices are within two steps precisely
when their supports intersect. This construction has \(2n\) vertices.


\paragraph{A lower bound from surjectivity.}
The key idea is to encode missing outputs of a function as orthogonal
pairs of rows.  Consider a surjectivity instance
\[
  f:[2n-2]\longrightarrow[\sqrt n]^2,
\]
where, for simplicity, assume \(n\) is a perfect square.  Write
\(f(t)=(f(t)_1,f(t)_2)\).  We form a Boolean matrix with
\(2\sqrt n\) rows and \(2n-2\) columns: for each output component and
each possible value, create a row whose entry in column \(t\) is one
exactly when that component of \(f(t)\) equals that value.

For example, take \(n=4\) and the function
\[
  \begin{array}{c|cccccc}
    t & 1 & 2 & 3 & 4 & 5 & 6\\ \hline
    f(t) & (1,1) & (1,2) & (2,1) & (1,1) & (1,2) & (2,1)
  \end{array},
\]
which misses \((2,2)\).  The corresponding matrix is
\[
  \begin{array}{c|cccccc}
    & 1 & 2 & 3 & 4 & 5 & 6\\ \hline
    f(t)_1=1 & 1 & 1 & 0 & 1 & 1 & 0\\
    f(t)_1=2 & 0 & 0 & 1 & 0 & 0 & 1\\
    f(t)_2=1 & 1 & 0 & 1 & 1 & 0 & 1\\
    f(t)_2=2 & 0 & 1 & 0 & 0 & 1 & 0
  \end{array}.
\]
A row for \(f(t)_1=u\) and a row for \(f(t)_2=v\) share a one precisely
in the columns \(t\) for which \(f(t)=(u,v)\).  These two rows are
therefore orthogonal exactly when \((u,v)\) is missing from the image of
\(f\).  In the example, the second and fourth rows form such a pair.

The full reduction in \Cref{thm:ov-lower-bound} uses this encoding to
combine \(\sqrt n\) independent surjectivity instances into a single
\(\OV\) instance with \(\Theta(n)\) rows and columns.  By the surjectivity lower bound of Beame and
Machmouchi~\cite{BeameMachmouchi2012} and standard adversary
composition~\cite{LeeEtAl2011}, deciding whether any of these functions
is not surjective requires \(\Omega(n\cdot \sqrt{\sqrt{n}})=\Omega(n^{5/4})\)
queries.  This gives \(Q(\OV_n)=\Omega(n^{5/4})\).

Combining \cref{eq:intro-bmpv-2} with this lower bound gives
\[
  Q(\OV_n)\le Q(\BMPV_n)
  =O\!\left(Q(\OV_{2n+2})\right),
\]
i.e., the two problems have the same quantum query complexity up to a possible
constant-factor change in problem dimension.

\subsubsection{An upper bound for Orthogonal Vectors}

We first explain the algorithm for \(\promOV_n\), in which every pair
among the input vectors \(x_1,\ldots,x_n\in\{0,1\}^n\) has either
\(O(\log n)\) or \(\Omega(n/\log n)\) common ones.
This promise still allows many pairs with few common ones, making
coordinate search costly, together with many vectors of large Hamming
weight, making it costly to learn their one positions before searching
for an orthogonal partner. Thus it retains the difficulties faced by
the two direct approaches giving the \(O(n^{3/2})\) upper bound.

We first describe the preprocessing, then the quantum walk under a
temporary assumption on individual list sizes.
We then remove this assumption and finally the intersection promise.
This overview is informal; the rigorous arguments appear in
\Cref{sec:ov-upper-bound}.

\paragraph{Preprocessing and combinatorial guarantees.}
Sampling \(O(\log^2 n)\) coordinates
and querying them in every vector reveals a common one for every pair
with a large intersection, with high probability. Let \(G\) be the graph on \([n]\) with an edge \(\{i,j\}\), for \(i\ne j\),
if no sampled coordinate is a common one of \(x_i\) and \(x_j\).
This graph is known after these queries; write \(N_G(i)\) for the
neighbors of \(i\). Every orthogonal pair is an edge of \(G\), and,
with high probability, every edge joins vectors with only
\(O(\log n)\) common ones.
We condition on this event in the promise case.

Choose independent parameters \(s,t\in[n]\). We sample \(s\) vector
indices independently and uniformly, call the corresponding vectors
pivots, and query them completely. We call \(j\) covered at coordinate
\(k\) if it is adjacent in \(G\) to some pivot having a one at \(k\).
Thus, if \(x_j(k)=1\) at a covered coordinate, this one is shared with
a sampled pivot. We call \(i\) heavy at \(k\) if more than \(t\) of
its neighbors in \(G\) remain uncovered at \(k\).

If \(i\) is heavy at \(k\) and \(x_i(k)=1\), we store \(k\) in a list
that can be reused when testing \(x_i\) against different vectors.
Accordingly, define
\(\Left{i}=\{k:x_i(k)=1,\ i\text{ is heavy at }k\}\).

We now argue that sampling the pivots makes the total size of the left lists
small. Fix a coordinate \(k\) and consider sampling the pivots one at a time. If many
heavy indices with a one remain, the next random pivot is likely to
be one of them. Whenever this happens, more than \(t\) new indices
become covered. But only \(n\) indices can become covered altogether.

Quantitatively, the number of heavy indices having a one can only
decrease. Each newly sampled pivot therefore covers,
in expectation, at least \(t/n\) times the expected number of such indices
remaining after all \(s\) pivot samples. This bounds the final expected number by \(n^2/(st)\). Summing
over coordinates gives \(\mathbb E\sum_i|\Left{i}|\le n^3/(st)\), and
hence, with constant probability, \(\sum_i|\Left{i}|=O(n^3/(st))\).

\paragraph{The quantum walk.}
We will perform a quantum walk on a Johnson graph to find a subset
containing an index of a vector with an orthogonal partner.
To explain the walk, temporarily assume the stronger statement that
every left list has size \(O(n^2/(st))\). We will remove this additional
assumption after describing the walk.

The walk maintains a subset \(I\subseteq[n]\) of size
\(r=\Theta(\sqrt n)\), together with the left lists of its indices. A
walk state is marked if some vector indexed by \(I\) has an orthogonal
partner anywhere in the input. Only one endpoint of an orthogonal
pair needs to be in the maintained subset.

An update replaces one index and loads its left list, keeping the
information about all other indices. Since heaviness is already known,
membership in a left list requires just one input query. Quantum
enumeration therefore gives update cost \(\Ot(\sqrt{n^3/(st)})\)
and setup cost \(\Ot(n^2/\sqrt{st})\) for loading all \(r\) lists.

If an orthogonal pair exists, a uniformly chosen \(r\)-element subset
\(I\subseteq[n]\) contains a fixed endpoint with probability \(r/n\).
Thus the marked fraction under the stationary distribution is at least
\(r/n\).
By the MNRS bound~\cite{MagniezNayakRolandSantha2011}
(\Cref{thm:mnrs}), with our choice of \(r\), the query cost of the quantum walk is
\(\Ot(n^2/\sqrt{st}+n^{1/4}C)\), where \(C\) is the cost of checking
a walk state.

We now describe the checking procedure. Fix the maintained subset \(I\) and a
possible partner \(j\). We first load the right list \(\Right{j}\),
consisting of the coordinates where \(x_j\) has a one and \(j\) is
covered. This list is small: every coordinate in it is a common one
of \(x_j\) and some pivot adjacent to \(j\). Each such pair has only
\(O(\log n)\) common ones, so \(|\Right{j}|=O(s\log n)\). Quantum
enumeration therefore loads it in \(\Ot(\sqrt{ns})\) queries. The right
list can now be reused for every candidate \(i\in I\cap N_G(j)\).

For a particular pair, let \(\Remaining{ij}\) contain the coordinates
where \(i\) is not heavy and \(j\) is uncovered. The search set is
\(\mathsf{Coords}_{ij}=\Left{i}\cup\Right{j}\cup\Remaining{ij}\).
Every common one belongs to this set:
heaviness puts it in the left list, coverage puts it in the right
list, and the remaining set handles the case where neither occurs.
Grover search over \(\mathsf{Coords}_{ij}\) therefore tests orthogonality using
\(\Ot(\sqrt{n^2/(st)+s+|\Remaining{ij}|})\) queries.

For each \(i\), we have
\(\sum_{j\in N_G(i)}|\Remaining{ij}|\le nt\). Indeed, at each
coordinate where \(i\) is not heavy, at most \(t\) neighbors of \(i\) in \(G\) are uncovered; summing over the \(n\) coordinates gives the bound.

Applying variable-time quantum search with these known cost bounds
over indices \(i\in I\cap N_G(j)\) and then over partners \(j\), we obtain
\[
 C=\Ot\!\left(n\sqrt s+
       \sqrt{nr\left(\frac{n^2}{st}+t\right)}\right).
\]

\paragraph{Removing the list size assumption.}
We now remove the assumption that every left list is small. The
bound on the total left list size limits the number of large lists.
We therefore group indices into size classes
according to approximate list sizes.

We sample \(\Theta((n/t)\log n)\) coordinates independently and query
them in every vector, using \(\Ot(n^2/t)\) additional queries.
From these samples, we estimate the left list sizes and assign each
index \(i\) a size label \(\mathsf{size}_i\) that is a power of two.
With high probability, one can show that each size label upper bounds its left list
size and satisfies \(\mathsf{size}_i=\Theta(|\Left{i}|+t)\).
Let the size class \(V_a\) contain the indices with size label \(a\),
and let \(n_a=|V_a|\). Using the bound on the total left list size,
it can be shown that \(a n_a=O(nt+n^3/(st))\).
Consequently, a size class with more expensive lists has fewer indices
to search.

We run the same walk on each size class, maintaining
\(r=\Theta(\sqrt{n_a})\) indices. Its query cost is
\(\Ot(\sqrt{n\,a\,n_a}+n\sqrt s\,n_a^{1/4})\).

There are only \(O(\log n)\) size classes. Including preprocessing, the
total cost is
\[
 \Ot\!\left(ns+\frac{n^2}{t}+\frac{n^2}{\sqrt{st}}
       +n\sqrt t+n^{5/4}\sqrt s\right).
\]
Choosing \(s=\lceil n^{1/3}\rceil\) and
\(t=\lceil n^{5/6}\rceil\) balances the three walk terms at
\(n^{17/12}\).

\paragraph{Removing the intersection promise.}
Without the promise, sampling a few coordinates no longer guarantees
that every pair with no sampled common one has a small intersection. Consequently, a
pivot may contribute many coordinates to a right list. We fix this by using only
the first common one of each pair, which suffices to rule out orthogonality.

Write \(\First{i}{j}\) for the coordinate of the first common one of
\(x_i\) and \(x_j\), or \(n+1\) if they are orthogonal. For each
\(k\in[n]\), define a graph \(G_k\) on \([n]\) with an edge
\(\{i,j\}\), for \(i\ne j\), exactly when \(k\le\First{i}{j}\).
Thus \(i\) and \(j\) are adjacent until and including their first common one.
We call \(j\) covered at \(k\) if it is adjacent in \(G_k\) to some pivot
having a one at \(k\). If \(x_j(k)=1\) and \(j\) is covered at \(k\),
then \(k\) must be the first common one of \(x_j\) and some pivot.
Each pivot contributes at most one such coordinate, restoring the
bound \(|\Right{j}|\le s\).

The graphs \(G_k\) are not known from preprocessing, so we estimate
the number of uncovered neighbors of each index as needed. We fully
query an independent sample of \(\Ot(n/t)\) vectors.
For each \(i,k\), the estimate is \(n\) times the fraction of sampled
indices that are adjacent to \(i\) in \(G_k\) and uncovered at \(k\).
We declare \(i\) heavy at \(k\) when this estimate exceeds \(2t\).
With high probability, every declared heavy index has more than \(t\)
such neighbors, while every index not declared heavy has at most \(4t\). We use these notions of coverage and heaviness in the same
definitions of the left lists, right lists, and remaining sets.
The counting argument for pivots still gives
\(\sum_i|\Left{i}|=O(n^3/(st))\), so the same size classes and
Johnson walks remain applicable.

The maintained data for a vector now consist of its left list and a
record containing the coordinates of its first common ones with the sampled vectors.
The record determines coverage, heaviness, and the right list. We
load it when needed in \(\Ot(\sqrt{n(s+n/t)})\) queries by searching
for successive coordinates that determine a previously unknown first
intersection. Each successful search determines the first common one with at least one sampled vector. There are \(\Ot(s+n/t)\) such searches, and their coordinate
intervals have total length at most \(n\).

One further change is needed in the checking procedure. The counting
argument now bounds only remaining coordinates up to the first
common one:
\[
 \sum_{j\ne i}
 \bigl|\{k\in\Remaining{ij}:k\le\First{i}{j}\}\bigr|\le4nt.
\]
We search increasingly long initial segments of \(\mathsf{Coords}_{ij}\)
in coordinate order, doubling their length each time, until we find
a common one or have searched the entire set. The cost is controlled by the prefix
ending at the first common one, or by the full search set if the pair
is orthogonal.

For a fixed partner \(j\), we load its record once and reuse it in
variable-time quantum search over \(i\in I\setminus\{j\}\), using
\Cref{lem:unknown-cost-search}.

These inner searches can themselves have variable stopping times.
We search over \(j\) using variable-time amplitude amplification
(\Cref{lem:variable-stopping-search}). Using the bounds on
the three parts of \(\mathsf{Coords}_{ij}\), one can show that checking
a walk state \(I\subseteq V_a\) of size \(r\) has query cost
\[
 C=\Ot\!\left(n\sqrt{s+n/t}+\sqrt{nra}\right).
\]

Finally, the size classes are now implicit. To test whether an index
belongs to \(V_a\), we load its record and query the coordinates used
to estimate its left list size. The subroutines for accessing implicit
sets then estimate \(n_a\) and select indices for preparing and
updating the walk. Accounting for these costs, one can show that the total query cost is
\[
 \Ot\!\left(ns+n\sqrt t+\frac{n^2}{\sqrt{st}}
       +n^{5/4}\sqrt s+\frac{n^{9/4}}{t}\right).
\]
The same choices \(s=\lceil n^{1/3}\rceil\) and
\(t=\lceil n^{5/6}\rceil\) give \(\Ot(n^{17/12})\) for the
general problem.

\subsubsection{Boolean Matrix Product Row Search and radius}
For every target \(v\), nested search decides \(\BMPRS_n\) with
\(O(n^{3/2})\) queries: Grover search over \(t\) evaluates
\((A\boolprod B)_{ij}\), a search over \(j\) looks for an entry of row
\(i\) that differs from \(v_j\), and a search over \(i\) looks for a row
with no such entry.  Search with bounded-error
inputs~\cite{HoyerMoscaDeWolf2003} avoids logarithmic factors.  We show
that this is optimal up to a factor \(\sqrt{\log n}\), already for
\(v=\ones_n\).

\paragraph{A lower bound from surjectivity.}
As for \(\OV\), we start from surjectivity, but now the quantifiers can be
matched directly.  A function \(f\colon[2q-2]\to[q]\) is not surjective
iff \(\exists x\,\forall j:f(j)\ne x\), which has the form of the two
outer quantifiers in \cref{eq:intro-exists-forall-exists}, with the rows
of \(A\) indexed by values \(x\) and the columns of \(B\) by positions
\(j\).  The inner quantifier must then express \(f(j)\ne x\).  We encode
each value $x$ as \(\gamma(x)\coloneqq(c(x),\overline{c(x)})\in\bits^\ell\),
where \(c(x)\) is the binary representation of \(x-1\) and
\(\ell=2\ceil{\log_2q}\).  Then \(\gamma(x)\) has a zero at some position
where \(\gamma(y)\) has a one if and only if \(x\ne y\): a bit in which
\(c(x)\) and \(c(y)\) differ gives such a position in one of the two
halves.  Hence, if row \(x\) of \(A\) indicates the zeros of \(\gamma(x)\)
and column \(j\) of \(B\) equals \(\gamma(f(j))\), then
\[
  (A\boolprod B)_{xj}
  =\bigvee_{b\in[\ell]}
   \bigl(\ind{\gamma(x)_b=0}\booland\gamma(f(j))_b\bigr)
  =\ind{f(j)\ne x},
\]
so row \(x\) of \(A\boolprod B\) is all ones exactly when \(f\) misses
\(x\).

This embeds surjectivity on \(q=\Theta(n)\) values, which requires
\(\Omega(n)\) queries~\cite{BeameMachmouchi2012}, but the inner dimension
is only \(\ell=O(\log n)\).  We expand it by replacing each column of
\(A\) with \(\Theta(n/\log n)\) identical columns and each bit in the
corresponding row of \(B\) with a block of fresh variables whose OR
represents that bit.  Adversary composition~\cite{LeeEtAl2011} with
\(\OR_{\Theta(n/\log n)}\) then gives \(\Omega(n\sqrt{n/\log n})\) queries,
even with \(A\) fixed; the encoding length causes the \(\sqrt{\log n}\)
loss.  Unlike the proof of \Cref{thm:ov-lower-bound}, which combines
\(\sqrt n\) surjectivity instances by an outer AND, a single instance
here occupies the two outer quantifiers, and the inner quantifier absorbs
an OR.

\paragraph{Radius.}
Taking both \(A\) and \(B\) to be the graph's adjacency matrix with ones
on the diagonal reduces Radius at Most Two to \(\BMPRS_n\) with
\(v=\ones_n\).  Conversely, a layered
construction~\cite{AbboudVassilevskaWilliamsWang2016} on \(3n+3\)
vertices has a row layer, an intermediate layer, and a column layer,
each with \(n\) vertices, joined according to \(A\) and \(B\), together
with three auxiliary vertices.  These put
each row vertex within distance two of every vertex other than the
column vertices, and they keep all other vertices from being centers.  A
row vertex is within distance two of the vertex for column \(j\) exactly
when they have a common neighbor in the intermediate layer, that is, when
\((A\boolprod B)_{ij}=1\).  Hence the graph has radius at most two
exactly when \(A\boolprod B\) has an all-ones row.

\subsection{Other related work}
\label{sec:related-work}

\paragraph{Boolean matrix products.}
Quantum algorithms for matrix products and verification are surveyed by
Kothari and Nayak~\cite{KothariNayak2016} and by Le
Gall~\cite{LeGall2014Survey}.  
Stronger lower bounds for linear-size depth-two circuits imply improved
lower bounds for BMPV~\cite{ChildsKimmelKothari2012,BunKothariThaler2021},
and Ding explored polynomial and span program approaches to Boolean
matrix problems~\cite{Ding2016}.  In classical fine-grained complexity,
Vassilevska Williams and Williams relate Boolean matrix multiplication,
triangle detection, and product verification through truly subcubic
combinatorial reductions~\cite{VassilevskaWilliamsWilliams2018}.

Quantum time--space tradeoffs for computing Boolean matrix products
were studied by Klauck, \v{S}palek, and de Wolf~\cite{KlauckSpalekDeWolf2007}
and strengthened by Beame, Kornerup, and
Whitmeyer~\cite{BeameKornerupWhitmeyer2026}.  These results concern
computation with a space constraint; we study decision problems and
charge only input queries.

\paragraph{Orthogonal Vectors.}


Orthogonal Vectors is a central source of classical fine-grained lower
bounds, with a broad equivalence class studied by Chen and
Williams~\cite{ChenWilliams2019}. Much of this classical literature focuses on
vectors of dimension $d=O(\log n)$. The regime $d=n$ has also been
studied: Dalirrooyfard and Kaufmann discuss high-dimensional
Orthogonal Vectors, allowing dimension up to $O(n)$, and its
connection to distinguishing diameter two from three in dense
graphs~\cite{DalirrooyfardKaufmann2021}.

In the quantum setting, 
Aaronson, Chia, Lin, Wang, and Zhang studied quantum Orthogonal Vectors
with entire-vector queries in polylogarithmic dimension and derived
conditional quantum time lower bounds~\cite{AaronsonEtAl2020}.
Our model instead charges for individual entries and uses dimension~$n$, so these quantum bounds are not directly comparable.
For Orthogonal Vectors in dimension $d=\omega(\log n)$,
Buhrman, Patro, and Speelman developed a QSETH framework connecting
black-box query lower bounds to conditional quantum time lower
bounds, and discussed its application to Orthogonal Vectors with
$d=\omega(\log n)$~\cite{BuhrmanPatroSpeelman2021}.
Van Renterghem obtained an unconditional $\Omega(n^{2/3})$
quantum query lower bound for Orthogonal Vectors with
$d=\omega(\log n)$ via a reduction from the two-to-one collision
problem~\cite{VanRenterghem2019}.
This lower bound applies, in particular, to the regime $d=n$
considered here.


\paragraph{Graph problems.}

Weso{\l}owski and Bao study quantum query algorithms and lower bounds
for eccentricity, radius, and diameter in the adjacency list
model~\cite{WesolowskiBao2025}.  
They give $\widetilde O(n\sqrt m)$-time quantum algorithms for
computing the exact diameter and radius, together with corresponding
paths, and prove $\Omega(\sqrt{nm})$ quantum query lower bounds,
where $m$ denotes the number of edges.
Our results concern unweighted graphs accessed through adjacency
matrix queries and the decision problems of testing whether diameter
or radius is at most two, so the bounds are not directly comparable.

\subsection{Open problems}

The main open problem is to close the gap between our
\(\Omega(n^{5/4})\) lower bound and \(\widetilde O(n^{17/12})\)
upper bound for Boolean Matrix Product Verification. Our reductions
show that this problem, Orthogonal Vectors, and deciding whether a
graph has diameter at most two have equivalent quantum query
complexity, up to constant factor changes in instance size. Thus,
improvements for any one of these problems will imply an improvement 
for all of them.
In particular, a lower bound of \(\Omega(n^{5/4+\varepsilon})\) for
some constant \(\varepsilon>0\) would establish a polynomial
separation from triangle finding, which admits an \(O(n^{5/4})\)
query algorithm \cite{LeGall2014Triangle,CaretteLauriereMagniez2017}.

Another direction is to understand the dependence on the vector
dimension. What is the quantum query complexity of Orthogonal
Vectors on \(n\) vectors of dimension \(d\), as a function of both
parameters? Our results address \(d=n\); it would be interesting
to obtain tight bounds when these parameters differ substantially,
as well as for verifying rectangular Boolean matrix products.

Our upper bound concerns query complexity and allows unrestricted
computation on acquired data. Can a quantum algorithm achieve a
running time comparable to our query bound?

Finally, does the polynomial separation between deciding whether
a graph has diameter at most two and whether it has radius at most
two persist when two is replaced by any larger fixed integer?

\paragraph{Organization.}
\Cref{sec:preliminaries} fixes the model and notation.
\Cref{sec:bmpv-and-ov} shows equivalences of Boolean Matrix Product
Verification, Orthogonal Vectors and the Diameter at Most Two 
problems and proves lower bounds for them.
\Cref{sec:ov-upper-bound} develops the upper
bound, first under an intersection promise and then for general inputs.
\Cref{sec:existential-row} gives lower bounds for Boolean Matrix 
Product Row Search and Radius at Most Two problems.
\section{Preliminaries}
\label{sec:preliminaries}

\subsection{Model and problem definitions}

For a positive integer \(n\), write \([n]\coloneqq\{1,\ldots,n\}\).
We write \(\J_{a,b}\) and \(\zeros_{a,b}\)
for the \(a\times b\) all-ones and all-zeros matrices, respectively.  We
abbreviate
\(\J_n\coloneqq\J_{n,n}\), write \(\ones_n\in\bits^n\) for the
all-ones vector, and write \(\I_n\) for the \(n\times n\) identity
matrix.  For \(X\in\bits^{a\times b}\), write
\(\overline X\coloneqq\J_{a,b}-X\) for its entrywise complement.
For any matrix \(X\), write \(X_{i,*}\) for its \(i\)-th row.
The Boolean product of \(A\in\bits^{a\times b}\) and
\(B\in\bits^{b\times c}\) is
\[
  (A\boolprod B)_{ij}
  \coloneqq\bigvee_{t\in[b]}(A_{it}\booland B_{tj}).
\]
For a Boolean vector \(x\), let
\(\supp(x)\coloneqq\{t:x(t)=1\}\).
For a proposition \(P\), \(\ind{P}\) denotes its indicator.

\paragraph{Query model and oracle reductions.}

A Boolean matrix \(X\in\bits^{a\times b}\) is accessed through the
standard entry oracle
\[
  O_X\ket{i,j,z}\coloneqq\ket{i,j,z\oplus X_{ij}}.
\]
For several input matrices or vector families, an additional register
selects the input object.  For a (possibly partial) Boolean function
\(f\from\mathcal D\to\bits\), we write \(Q(f)\) for the minimum number of
queries to the input oracle made by a quantum algorithm that, for every
\(x\in\mathcal D\), outputs \(f(x)\) with probability at least \(2/3\).
This is its bounded-error quantum query complexity; in particular,
\(Q(1-f)=Q(f)\).

An oracle reduction has \emph{constant query overhead} if its target
oracle can be simulated coherently using \(O(1)\) input queries per target
query.  Unless stated otherwise, every reduction in this paper has this
property.

All computation on acquired data and all
input-independent unitaries are free in this model; our upper bounds
concern query complexity.

\Needspace{6\baselineskip}
\paragraph{Orthogonal Vectors.}

\begin{definition}[Orthogonal Vectors]
\label{def:ov}
Given an indexed list \((x_1,\ldots,x_m)\) of vectors in \(\bits^d\),
\(\OV_{m,d}\) accepts if there are distinct indices \(i,j\in[m]\)
with
\[
  \supp(x_i)\cap\supp(x_j)=\varnothing.
\]
We call such a pair orthogonal and abbreviate
\(\OV_n=\OV_{n,n}\).  Equal bit strings at different
positions are distinct choices.
\end{definition}

Orthogonal Vectors is also commonly formulated with two lists, asking
for one vector from each list with disjoint supports~\cite{ChenWilliams2019}.
The following elementary reductions relate this convention to
\Cref{def:ov}.

\begin{lemma}[Equivalent formulations of Orthogonal Vectors]
\label{lem:ov-formulations}
Two lists of \(m\) vectors in \(\bits^d\) reduce to
\(\OV_{2m,d+2}\), and \(\OV_{m,d}\) reduces to two
lists of \(m\) vectors in dimension \(d+m\).  Each queried target bit
is fixed or copies one input bit.
\end{lemma}

\begin{proof}
View the vectors as rows of a matrix.  For the first reduction, combine
the two lists and append \((1,0)\)
to every vector in the first list and \((0,1)\) to every vector in the
second.  The new columns exclude pairs from the same list and
preserve every cross-list intersection test.
For the reverse reduction, make two copies of the input list and append
the \(i\)-th unit vector of \(\bits^m\) to position \(i\) in both copies.
The added columns make the two copies of position \(i\) intersect,
and create no intersection between distinct positions.
\end{proof}

Hence bounds for list lengths and dimensions within constant factors of
\(n\) transfer between the formulations.  We also use the following
padding facts.

\begin{remark}[Padding]
\label{rem:ov-padding}
Appending all-zero columns to the matrix of input vectors preserves the
answer to \(\OV\).  So does appending copies of vectors already in the list,
provided the list has at least two positions: two copies of the same
vector are orthogonal only if that vector is zero, and a zero vector is
already orthogonal to every other vector in the list.
\end{remark}

For \(X\in\bits^{n\times d}\) with \(n\ge2\),
\begin{equation}
  \OV(X)=0
  \quad\Longleftrightarrow\quad
  X\boolprod X^\transpose=\J_n.
  \label{eq:no-ov-product}
\end{equation}
The off-diagonal product entries express pairwise intersection.
If no two distinct rows are orthogonal, every row is nonzero, so the
diagonal entries are also one.

\paragraph{Boolean product and graph predicates.}

\begin{definition}[Boolean Matrix Product Verification]
\label{def:bmpv}
For \(A,B,C\in\bits^{n\times n}\),
\[
  \BMPV_n(A,B,C)\coloneqq\ind{A\boolprod B=C}.
\]
We also use
\[
  \BMPV_n^{\le}(A,B,C)\coloneqq
  \ind{A\boolprod B\le C},
  \qquad
  \BMPV_n^{\ge}(A,B,C)\coloneqq
  \ind{A\boolprod B\ge C}.
\]
\end{definition}

For a simple undirected \(n\)-vertex graph \(G\), let \(A_G\) be its
zero-diagonal adjacency matrix.  Define its closed neighborhood matrix by
\[
  \widehat A_G\coloneqq\I_n\boolor A_G,
\]
where the OR is entrywise.  The input oracle supplies the off-diagonal
entries of \(A_G\); symmetry is promised.

For \(u,v\in V(G)\), let \(\operatorname{dist}_G(u,v)\) be the length of a
shortest path between \(u\) and \(v\), or \(\infty\) if they lie in different
components.
For \(n\)-vertex graphs, define
\[
  \DIAM_{\le2,n}(G)\coloneqq\ind{\operatorname{diam}(G)\le2},
  \qquad
  \RAD_{\le2,n}(G)\coloneqq\ind{\operatorname{rad}(G)\le2}.
\]
The basic identity
\begin{equation}
  (\widehat A_G\boolprod\widehat A_G)_{uv}=1
  \quad\Longleftrightarrow\quad
  \operatorname{dist}_G(u,v)\le2
  \label{eq:closed-neighborhood-square}
\end{equation}
will be used in both graph applications.  In particular,
\begin{equation}
  \DIAM_{\le2,n}(G)=1
  \quad\Longleftrightarrow\quad
  \widehat A_G\boolprod\widehat A_G=\J_n.
  \label{eq:diameter-closed-square}
\end{equation}

Finally, \(\Triangle_n\) denotes triangle detection in an \(n\)-vertex
graph under adjacency matrix access.

\subsection{Quantum algorithmic tools}
\label{sec:quantum-tools}

\paragraph{Search and access to implicit sets.}
We collect the search, enumeration, and walk procedures used in the
upper bounds. All index sets and numerical bounds below have size
polynomial in the input parameter \(n\). Within \(\widetilde O\)
bounds, subroutines are amplified to inverse polynomial error,
implemented coherently, and stopped at the stated query bounds when
these are known.
Grover search and amplitude amplification find a marked element of
\([M]\), or report that none exists, using \(O(\sqrt M)\) predicate
evaluations~\cite{BrassardHoyerMoscaTapp2002}. We use the following
standard consequences of amplitude estimation and quantum
enumeration~\cite{BrassardHoyerMoscaTapp2002,JefferyKothariMagniez2012}.

\begin{lemma}[Access to an implicit set]
\label{lem:implicit-set-access}
Let \(A\subseteq[M]\), with \(m=|A|\), have a membership test costing
\(q\) input queries. The following tasks can be performed with
inverse polynomial error:
\begin{enumerate}
\item Distinguish \(m=0\) from \(m>0\), and in the latter case obtain
      \(\widehat m\in[m/2,2m]\), in \(\widetilde O(q\sqrt M)\) queries.
\item Given such an estimate for \(m>0\), prepare the uniform
      superposition over \(A\), to inverse polynomial accuracy, in
      \(\widetilde O(q\sqrt{M/m})\) queries.
\item Given \(m\le K\), enumerate \(A\) in
      \(\widetilde O(q\sqrt{M(K+1)})\) queries.
\end{enumerate}
\end{lemma}

\paragraph{Searching in coordinate order.}
If the first marked position in an ordered domain of length \(M\) is
at position \(g\), it can be located in \(\widetilde O(\sqrt g)\)
predicate evaluations. Search consecutive intervals of doubling
length until one contains a mark, then locate its first mark by binary
subdivision. If no position is marked, return an end marker at
\(M+1\), with \(g=M+1\).

\begin{lemma}[Successive marked positions]
\label{lem:successive-discoveries}
Consider a procedure that searches for the next marked position after
a cursor in \([M]\), advances the cursor to that position, and stops
after at most \(K\) marked positions or upon reaching the end marker.
The predicate may change after each marked position, provided previously
passed positions need not be revisited. If a predicate evaluation
uses \(O(1)\) queries, the entire procedure costs
\(\widetilde O(\sqrt{(M+1)(K+1)})\) queries.
\end{lemma}

\begin{proof}
If the successive gaps are \(g_1,\ldots,g_h\), then
\(h\le K+1\) and \(\sum_{v=1}^h g_v\le M+1\).
The searches cost \(\widetilde O(\sum_v\sqrt{g_v})\) queries.
Since \(\sum_v\sqrt{g_v}\le\sqrt{h\sum_vg_v}\), this gives
the claimed bound.
\end{proof}

\paragraph{Variable-time search.}
The cost of testing a candidate may depend on the candidate and on
the input. We use Ambainis's variable-time search algorithm for unknown
completion costs~\cite{Ambainis2010VariableTimes}, with the usual
bounded-error extension. Neither the individual completion costs nor
their sum of squares needs to be known in advance.

\begin{lemma}[Search with unknown completion costs]
\label{lem:unknown-cost-search}
For each candidate \(j\), suppose a test with integer allowance
\(u\ge1\) costs \(\widetilde O(u)\) queries and returns
\textnormal{\textsc{Unfinished}} for \(u<\tau_j\), and the correct
Boolean answer for \(u\ge\tau_j\), with inverse polynomial error.
There is a variable-time quantum search algorithm that, with
probability at least \(1-n^{-c}\), finds an accepting candidate or
correctly reports that none exists, and finishes within
\(\widetilde O(\sqrt{\sum_j\tau_j^2})\) queries, for any specified
constant \(c>0\). The algorithm does not require the values
\(\tau_j\) or a bound on their sum of squares.
\end{lemma}

In particular, suppose test \(j\) costs
\(\widetilde O(\sqrt{w_j})\) queries, where \(w_j\ge1\) is known
without further input queries. Taking
\(\tau_j=\lceil\sqrt{w_j}\rceil\), and returning
\textsc{Unfinished} below this allowance, gives a search cost of
\(\widetilde O(\sqrt{\sum_jw_j})\) queries.

The search above can itself be used to test a candidate in another
search. Its stopping time may depend on intermediate measurement
outcomes. For this purpose, we use the following consequence of
variable-time amplitude amplification~\cite[Theorem~1]{Ambainis2012VariableTimeAmplification}.

\begin{lemma}[Search with variable stopping times]
\label{lem:variable-stopping-search}
Suppose each candidate \(j\in[n]\) has a Boolean test that, with
high probability, returns the correct answer within \(q_j\ge1\)
queries. Given a bound \(Q\ge\sqrt{\sum_{j=1}^n q_j^2}\), quantum search
finds an accepting candidate or correctly reports that none exists
in \(\widetilde O(Q)\) queries, with high probability.
The individual bounds \(q_j\) need not be known.
\end{lemma}

\begin{proof}
Stop each test after \(\lceil Q\rceil\) queries, returning
\textsc{No} if it has not finished. Since \(q_j\le Q\), this
preserves correctness with high probability. For a uniform candidate,
the mean squared stopping time is \(O(Q^2/n)\): runs finishing
within their bounds contribute at most \(n^{-1}\sum_j q_j^2\), and exceptional runs
contribute at most their probability times \(\lceil Q\rceil^2\).
Here the failure probability is taken sufficiently small.
If an accepting candidate exists, the success probability is
\(\Omega(1/n)\). Variable-time amplitude amplification therefore
finds one in \(\widetilde O(Q)\) queries. Stop at this bound and
report that none exists if no accepting candidate is found.
The error convention above controls incorrect test outcomes,
and repetition gives high probability.
\end{proof}


\paragraph{MNRS quantum walk search.}
The framework of Magniez, Nayak, Roland, and Santha expresses the cost
of a quantum walk through setup, update, and checking
costs~\cite{MagniezNayakRolandSantha2011}. We state the form used
below for a Johnson walk, each of whose steps replaces one
uniform member of the maintained subset by a uniform element outside it.

\begin{theorem}[MNRS search on a Johnson graph]
\label{thm:mnrs}
Consider the Johnson walk on the \(r\)-subsets of a set of size
\(n\ge3\), where \(1\le r\le n/2\). Suppose that, whenever marked
walk states exist, every subset containing a particular element is
marked. If preparing a uniform walk state with its data costs \(S\)
queries, a coherent walk update costs \(U\), and checking whether a
walk state is marked costs \(C\), then detecting a marked walk state uses
\begin{equation}
 \widetilde O\!\left(S+\sqrt n\,U+\sqrt{n/r}\,C\right)
 \label{eq:johnson-search}
\end{equation}
queries.
\end{theorem}

The walk has a uniform stationary distribution and spectral gap
\(\Theta(1/r)\). When marked walk states exist, their stationary
probability is at least \(r/n\), giving the stated bound.

\section{From verification to Orthogonal Vectors}
\label{sec:bmpv-and-ov}

A claimed product \(C\) can differ from \(A\boolprod B\) in two ways.
Entry \((i,j)\) has a \emph{missing-one error} if
\((A\boolprod B)_{ij}=1\) and \(C_{ij}=0\), and a \emph{spurious-one
error} if \((A\boolprod B)_{ij}=0\) and \(C_{ij}=1\).  Thus
\(\BMPV_n^{\le}\) rejects exactly when some entry has a missing-one error,
\(\BMPV_n^{\ge}\) rejects exactly when some entry has a spurious-one error,
and \(\BMPV_n=\BMPV_n^{\le}\booland\BMPV_n^{\ge}\).

The two kinds of error behave very differently.  Missing-one errors are
triangles in disguise~\cite[Section~4.3]{Kothari2014}, so they can be
detected in \(O(n^{5/4})\) queries.  Spurious-one errors are orthogonal
pairs in disguise, and we show in \Cref{sec:bmpv-reductions} that
detecting them is equivalent to Orthogonal Vectors.  In
\Cref{sec:ov-lower-bound}, we prove an \(\Omega(n^{5/4})\) lower bound for
Orthogonal Vectors.  This lower bound absorbs the cost of triangle
detection, so \(\BMPV\) and \(\OV\) have the same query complexity up to
a constant-factor change in problem dimension (\Cref{sec:bmpv-equivalence}).
\Cref{sec:diameter} proves the same equivalence for Diameter at Most Two.
Through these reductions, the algorithm for Orthogonal Vectors in
\Cref{sec:ov-upper-bound} also applies to verification and to Diameter at
Most Two.

\subsection{Reductions between verification and Orthogonal Vectors}
\label{sec:bmpv-reductions}

A missing-one error at \((i,j)\), together with an index \(t\) satisfying
\(A_{it}=B_{tj}=1\), forms a triangle in a tripartite graph built from
\(A\), \(B\), and \(\overline C\).  Kothari observed that this
correspondence works in both directions.

\begin{proposition}[Missing-one errors and Triangles~\cite{Kothari2014}]
\label{prop:triangle-half}
The complement of \(\BMPV_n^{\le}\) reduces to triangle detection on
\(3n\) vertices, and triangle detection on \(n\) vertices reduces to the
complement of \(\BMPV_n^{\le}\).  Consequently,
\[
  Q(\Triangle_n)
  \le Q(\BMPV_n^{\le})
  \le Q(\Triangle_{3n})
  =O(n^{5/4}).
\]
\end{proposition}

\begin{proof}
Given \(A,B,C\in\bits^{n\times n}\), consider the graph on three disjoint
vertex sets \(U=\{u_i\}\), \(K=\{k_t\}\), and \(V=\{v_j\}\), each of
size \(n\), with adjacency matrix
\begin{equation}
  \begin{bmatrix}
    0 & A & \overline C\\
    A^\transpose & 0 & B\\
    \overline C^\transpose & B^\transpose & 0
  \end{bmatrix}.
  \label{eq:triangle-block}
\end{equation}
Each adjacency query costs at most one query to \(A\), \(B\), or \(C\).
The graph is tripartite, so every triangle has the form
\((u_i,k_t,v_j)\), and such a triangle exists exactly when
\(A_{it}=B_{tj}=1\) and \(C_{ij}=0\).  Thus the graph contains a triangle
if and only if \(C\) has a missing-one error.

Conversely, let \(D\) be the adjacency matrix of an \(n\)-vertex graph.
An entry with \((D\boolprod D)_{ij}=1\) and \(\overline D_{ij}=0\) is an
edge \(\{i,j\}\) whose endpoints have a common neighbor, which is exactly
a triangle.  Diagonal entries cause no violations, because
\(\overline D_{ii}=1\).  Hence the graph contains a triangle if and only
if \((D,D,\overline D)\) is a no-instance of \(\BMPV_n^{\le}\), and each
entry of this triple costs one query to \(D\).

The final bound is the triangle-finding algorithm of Carette,
Lauri\`ere, and Magniez~\cite{CaretteLauriereMagniez2017}, which removes
the logarithmic factors from Le Gall's
\(\widetilde O(n^{5/4})\) algorithm~\cite{LeGall2014Triangle}.
\end{proof}

Spurious-one errors are governed by orthogonality instead: the entry
\((A\boolprod B)_{ij}\) is zero exactly when row \(i\) of \(A\) and
column \(j\) of \(B\) have disjoint supports.  A spurious-one error is
such an orthogonal pair at a position with \(C_{ij}=1\).  The next lemma
encodes this restriction into a single \(\OV\) instance and shows,
conversely, that \(\OV\) is the special case
\((A,B,C)=(X,X^\transpose,\J_n)\).

\begin{lemma}[Spurious-one errors and Orthogonal Vectors]
\label{lem:bmpv-ge-ov}
For every \(n\ge2\),
\begin{equation}
  Q(\OV_n)
  \le Q(\BMPV_n^{\ge})
  \le Q(\OV_{2n+2}).
  \label{eq:bmpv-ge-ov}
\end{equation}
\end{lemma}

\begin{proof}
For the second inequality, let \(Z\) be the matrix
\begin{equation}
  Z\coloneqq
  \begin{bmatrix}
    A & \I_n & \J_{n,1} & \zeros_{n,1}\\
    B^\transpose & \overline C^\transpose
      & \zeros_{n,1} & \J_{n,1}\\
    \J_{2,n} & \J_{2,n} & \J_{2,1} & \J_{2,1}
  \end{bmatrix}
  \in\bits^{(2n+2)\times(2n+2)}.
  \label{eq:bmpv-ov-blocks}
\end{equation}
Let \(e_i\in\bits^n\) denote the \(i\)-th standard basis vector.
The first \(n\) rows of \(Z\) are \(a_i\coloneqq(A_{i,*},e_i,1,0)\) for
\(i\in[n]\), and its next \(n\) rows are
\(b_j\coloneqq\bigl((B^\transpose)_{j,*},(\overline C^\transpose)_{j,*},0,1\bigr)\)
for \(j\in[n]\).  In the first \(n\) columns, \(a_i\) and \(b_j\)
share a one exactly when \((A\boolprod B)_{ij}=1\).  In the next \(n\)
columns, \(a_i\) has a single one, in column \(n+i\), where \(b_j\)
has the entry \(\overline C_{ij}\); so they share a one there exactly when
\(C_{ij}=0\).  They share no one in the last two columns.  Therefore
\(a_i\) and \(b_j\) are orthogonal exactly when entry \((i,j)\) has a
spurious-one error.

No other pair of rows of \(Z\) is orthogonal: any two rows \(a_i\) share
column \(2n+1\), any two rows \(b_j\) share
column \(2n+2\), and the last two rows are all ones, while every row
of \(Z\) is nonzero.  These two rows serve only to make the instance
square.  Hence the rows of \(Z\) form a yes-instance of \(\OV_{2n+2}\)
if and only if \(A\boolprod B\not\ge C\).  Every entry of \(Z\) is either
fixed or determined by one entry of \(A\), \(B\), or \(C\).

For the first inequality, let \(X\in\bits^{n\times n}\) be an \(\OV_n\)
instance whose rows are the input vectors.  Since \(A\boolprod B\le\J_n\)
always holds, \Cref{eq:no-ov-product} gives
\[
  \BMPV_n^{\ge}(X,X^\transpose,\J_n)=1
  \quad\Longleftrightarrow\quad
  X\boolprod X^\transpose=\J_n
  \quad\Longleftrightarrow\quad
  \OV_n(X)=0.
\]
Queries to \(X^\transpose\) are queries to \(X\), and \(\J_n\) is fixed.
\end{proof}

\subsection{A lower bound from surjectivity}
\label{sec:ov-lower-bound}

We now prove that \(\OV_n\) requires \(\Omega(n^{5/4})\) queries.  The
starting point is surjectivity.  For even \(q\), let \(\ONTO_q\) be the
problem of deciding whether a function
\[
  f\from[2q-2]\longrightarrow[q],
\]
given by an oracle that returns \(f(t)\) on query \(t\), is surjective.
Beame and Machmouchi proved that \(Q(\ONTO_q)=\Omega(q)\) for even
\(q\)~\cite[Corollary~6]{BeameMachmouchi2012}.  We need the corresponding
bound for many independent instances.

\begin{lemma}[An AND of \(\ONTO\) instances]
\label{lem:and-of-onto}
For every even \(q\) and every \(r\ge1\), deciding whether \(r\)
independently given functions \([2q-2]\to[q]\) are all surjective requires
\(\Omega(q\sqrt r)\) quantum queries.
\end{lemma}

\begin{proof}
Let \(\Adv\) denote the negative-weight adversary bound.  It
characterizes bounded-error quantum query complexity up to a constant
factor and is multiplicative under composition, including for functions
over finite input alphabets~\cite[Theorem~1.1 and
Lemma~5.2]{LeeEtAl2011}.  Hence \(\Adv(\ONTO_q)=\Omega(q)\), and since
\(\Adv(\AND_r)=\sqrt r\),
\[
  Q\!\left(\AND_r\circ\ONTO_q^{\,r}\right)
  =\Omega\!\left(\Adv(\AND_r)\,\Adv(\ONTO_q)\right)
  =\Omega(q\sqrt r).
  \qedhere
\]
\end{proof}

For intuition, take \(r\approx\sqrt n\) functions whose values are pairs,
with about \(\sqrt n\) possibilities for each component.  The alphabet
then has size \(q\approx n\).  The key observation is that
a function \(f\) misses a value \((u,v)\) exactly when the set of
positions at which the first component of \(f\) equals \(u\) is disjoint
from the set of positions at which the second component equals \(v\).
Each such pair of sets becomes a pair of rows, and a few additional
columns ensure that no other pair of rows is orthogonal.

\begin{theorem}[Lower bound for Orthogonal Vectors]
\label{thm:ov-lower-bound}
For every \(n\ge2\),
\[
  Q(\OV_n)=\Omega(n^{5/4}).
\]
\end{theorem}

\begin{proof}
Let \(s\) be even, and consider \(s\) functions
\[
  f_a\from[2s^2-2]\longrightarrow[s]^2,
  \qquad a\in[s],
\]
where a query \((a,t)\) returns \(f_a(t)\).  By \Cref{lem:and-of-onto}
with \(q=s^2\) and \(r=s\), deciding whether every \(f_a\) is surjective
requires \(\Omega(s^{5/2})\) queries.  We reduce the complementary
problem, deciding whether some \(f_a\) misses a value, to \(\OV\).

\paragraph{The vectors.}
For \(i\in\{1,2\}\), let \(\pi_i\) map a pair to its \(i\)-th component.
The instance consists of the \(2s^2\) vectors
\[
  x^{(i)}_{a,u},
  \qquad i\in\{1,2\},\ a,u\in[s],
\]
arranged as the rows of a matrix with three types of columns.
\begin{itemize}
\item \emph{Value columns.}  For each \(t\in[2s^2-2]\),
  \(x^{(i)}_{a,u}(t)\coloneqq\ind{\pi_i(f_a(t))=u}\).
\item \emph{Guard columns.}  For \(j\in\{1,2\}\),
  \(x^{(i)}_{a,u}(g_j)\coloneqq\ind{i\ne j}\).
\item \emph{Index columns.}  Let \(\ell\coloneqq\ceil{\log_2s}\), and
  let \(c(a)\in\bits^\ell\) be the binary representation of \(a-1\).  For
  each \(b\in[\ell]\) and \(\sigma\in\{(0,1),(1,0)\}\),
  \(x^{(i)}_{a,u}(h_{b,\sigma})\coloneqq\ind{c(a)_b=\sigma_i}\).
\end{itemize}
\Cref{fig:ov-explicit-construction} shows the complete instance for
\(s=2\).

\begin{figure}[ht!]
\centering
\small
{\setlength{\tabcolsep}{5pt}
\begin{tabular}{c cccccc}
\toprule
\(t\) & \(1\) & \(2\) & \(3\) & \(4\) & \(5\) & \(6\) \\
\midrule
\(f_1(t)\) & \((1,1)\) & \((1,2)\) & \((2,1)\) & \((1,1)\)
  & \((1,2)\) & \((2,1)\) \\
\(f_2(t)\) & \((1,1)\) & \((1,2)\) & \((2,1)\) & \((2,2)\)
  & \((1,1)\) & \((2,2)\) \\
\bottomrule
\end{tabular}\par}

\medskip

{\setlength{\tabcolsep}{3.5pt}
\begin{tabular}{c cccccc cc cc}
\toprule
& \multicolumn{6}{c}{\(t\)}
& \multicolumn{2}{c}{\(g_j\)}
& \multicolumn{2}{c}{\(h_{1,\sigma}\)} \\
\cmidrule(lr){2-7}\cmidrule(lr){8-9}\cmidrule(lr){10-11}
& \(1\) & \(2\) & \(3\) & \(4\) & \(5\) & \(6\)
& \(g_1\) & \(g_2\) & \(h_{1,(0,1)}\) & \(h_{1,(1,0)}\) \\
\midrule
\(x^{(1)}_{1,1}\) & 1&1&0&1&1&0 & 0&1 & 1&0 \\
\(\boldsymbol{x}^{(1)}_{1,2}\) & 0&0&1&0&0&1 & 0&1 & 1&0 \\
\(x^{(2)}_{1,1}\) & 1&0&1&1&0&1 & 1&0 & 0&1 \\
\(\boldsymbol{x}^{(2)}_{1,2}\) & 0&1&0&0&1&0 & 1&0 & 0&1 \\
\addlinespace
\(x^{(1)}_{2,1}\) & 1&1&0&0&1&0 & 0&1 & 0&1 \\
\(x^{(1)}_{2,2}\) & 0&0&1&1&0&1 & 0&1 & 0&1 \\
\(x^{(2)}_{2,1}\) & 1&0&1&0&1&0 & 1&0 & 1&0 \\
\(x^{(2)}_{2,2}\) & 0&1&0&1&0&1 & 1&0 & 1&0 \\
\bottomrule
\end{tabular}\par}
\caption{The construction in the proof of \Cref{thm:ov-lower-bound} for
\(s=2\).  The function \(f_1\) misses \((2,2)\), whereas \(f_2\) is surjective.
The two vectors with boldface labels have disjoint supports and form the
only orthogonal pair.  Two vectors with the same \(i\) intersect in a
guard column, and two vectors with different \(i\) and different \(a\)
intersect in an index column.}
\label{fig:ov-explicit-construction}
\end{figure}

\paragraph{Correctness.}
Vectors with \(i=1\) all have a one at \(g_2\), and vectors with \(i=2\)
all have a one at \(g_1\), so any two vectors with the same \(i\)
intersect.  Consider therefore \(x^{(1)}_{a,u}\) and \(x^{(2)}_{a',v}\).
These vectors share no one in a guard column.  If \(a\ne a'\), choose \(b\)
with \(c(a)_b\ne c(a')_b\), and let \(\sigma\coloneqq(c(a)_b,c(a')_b)\).
Then \(\sigma\in\{(0,1),(1,0)\}\), and both vectors have a one at
\(h_{b,\sigma}\).  If \(a=a'\), a common one at \(h_{b,\sigma}\) would
require \(c(a)_b=\sigma_1\) and \(c(a)_b=\sigma_2\), which is impossible,
so the two vectors share no one in an index column.  In the value columns,
they share a one at \(t\) exactly when \(f_a(t)=(u,v)\).  Consequently,
the only orthogonal pairs are the pairs \(x^{(1)}_{a,u},x^{(2)}_{a,v}\)
with \((u,v)\notin\im(f_a)\), and the instance contains an orthogonal pair
if and only if some \(f_a\) is not surjective.  Every entry is either fixed or
determined by a single value \(f_a(t)\), so each query to the instance can be simulated using \(O(1)\) queries to the functions.

\paragraph{Size.}
The instance consists of \(2s^2\) vectors of dimension
\[
  d(s)\coloneqq2s^2+2\ceil{\log_2s}=\Theta(s^2).
\]
Given a sufficiently large \(n\), let \(s\) be the largest even integer
with \(d(s)\le n\); then \(s^2=\Theta(n)\).  Pad the list to \(n\) vectors
by duplicating vectors, and pad the dimension to \(n\) with all-zero
columns.  By \Cref{rem:ov-padding}, padding preserves the
answer.  Hence
\[
  Q(\OV_n)=\Omega(s^{5/2})=\Omega(n^{5/4}).
  \qedhere
\]
\end{proof}

\subsection{Equivalence of verification and Orthogonal Vectors}
\label{sec:bmpv-equivalence}

We now combine the preceding results.

\begin{theorem}[Verification and Orthogonal Vectors]
\label{thm:bmpv-final}
For every \(n\ge2\),
\begin{equation}
  Q(\OV_n)
  \le Q(\BMPV_n)
  =O\!\left(Q(\OV_{2n+2})\right).
  \label{eq:bmpv-structural}
\end{equation}
Consequently, \(Q(\BMPV_n)=\Omega(n^{5/4})\).
\end{theorem}

\begin{proof}
For the first inequality, restrict \(\BMPV_n\) to inputs of the form
\((X,X^\transpose,\J_n)\).  By \Cref{eq:no-ov-product}, this restriction
is equivalent to the complement of \(\OV_n\).  Therefore
\(Q(\OV_n)\le Q(\BMPV_n)\).

For the second relation, decide \(\BMPV_n^{\le}\) using
\Cref{prop:triangle-half} and \(\BMPV_n^{\ge}\) using
\Cref{lem:bmpv-ge-ov}, each with error probability at most \(1/6\), and
accept if both accept.  The triangle algorithm detects missing-one
errors, and the \(\OV\) algorithm detects spurious-one errors.  This
gives
\[
  Q(\BMPV_n)
  =O\!\left(Q(\Triangle_{3n})+Q(\OV_{2n+2})\right)
  =O\!\left(n^{5/4}+Q(\OV_{2n+2})\right).
\]
By \Cref{thm:ov-lower-bound}, \(Q(\OV_{2n+2})=\Omega(n^{5/4})\), so the
second term dominates, proving \Cref{eq:bmpv-structural}.

Finally, the lower bound \(Q(\BMPV_n)=\Omega(n^{5/4})\) follows from
the first inequality in \Cref{eq:bmpv-structural} and
\Cref{thm:ov-lower-bound}.
\end{proof}

By \Cref{eq:no-ov-product}, testing whether
\(X\boolprod X^\transpose=\J_n\) is the complement of \(\OV_n\).
\Cref{thm:bmpv-final} therefore shows that this self-product test, a
special case of verification with \(B=A^\transpose\) and the all-ones
target, is as hard as verifying an arbitrary claimed product, up to
replacing \(n\) by \(2n+2\).

\subsection{Diameter at Most Two}
\label{sec:diameter}

By \Cref{eq:closed-neighborhood-square}, two vertices are at distance at
most two exactly when their closed neighborhoods intersect.  A graph
therefore has diameter at most two exactly when the rows of
\(\widehat A_G\), viewed as an \(\OV\) instance, contain no orthogonal
pair.  The next proposition turns this observation into reductions in both
directions.

\begin{proposition}[Diameter and Orthogonal Vectors]
\label{prop:diameter-ov}
For every \(n\ge2\),
\[
  Q(\DIAM_{\le2,n})
  \le Q(\OV_n)
  \le Q(\DIAM_{\le2,2n}).
\]
In particular, $Q(\DIAM_{\le2,n})=\Omega(n^{5/4}).$
\end{proposition}

\begin{proof}
For the first inequality, let \(G\) be an \(n\)-vertex graph, and use the
rows of \(\widehat A_G\) as the input list.  The vector for vertex \(v\)
has a one in column \(w\) if \(v=w\), and otherwise has the adjacency
bit \((A_G)_{vw}\).  Thus each entry query requires at most one adjacency
query.  By \Cref{eq:diameter-closed-square} and \Cref{eq:no-ov-product},
\(G\) has diameter at most two exactly when the list has no orthogonal pair.

For the second inequality, let \(X\in\bits^{n\times n}\) be an \(\OV_n\)
instance.  Define a graph \(H_X\) on \(2n\) vertices with two disjoint
vertex sets
\[
  R\coloneqq\{r_1,\ldots,r_n\}
  \qquad\text{and}\qquad
  W\coloneqq\{w_1,\ldots,w_n\},
\]
where \(r_i\) represents row \(i\) of \(X\) and \(w_t\) represents
column \(t\).  The vertices in \(R\) are independent, and the vertices
in \(W\) form a clique.  A vertex \(r_i\) is adjacent to \(w_t\) exactly
when \(X_{it}=1\).  Each adjacency query can therefore be answered using
at most one query to \(X\).

Since the vertices in \(W\) form a clique, every vertex in \(R\) that has
a neighbor is within distance two of every vertex in \(W\).  Two distinct
vertices \(r_i,r_j\in R\) are within distance two exactly when they have
a common neighbor in \(W\), so
\[
  \operatorname{dist}_{H_X}(r_i,r_j)\le2
  \quad\Longleftrightarrow\quad
  \text{\(X_{it}=X_{jt}=1\) for some \(t\in[n]\)}.
\]
If row \(i\) of \(X\) is zero, then \(r_i\) is isolated in \(H_X\).
Since \(n\ge2\), this row forms an orthogonal pair with another row.
Hence \(H_X\) has diameter at most two exactly when \(X\) has no
orthogonal pair.

For the lower bound, small values of \(n\) affect only the constant,
so assume \(n\ge4\).  Let \(m\coloneqq\floor{n/2}\), and apply the second
reduction above to an \(\OV_m\) instance.  When \(n\) is odd, add
one more vertex to the clique \(W\) with no neighbors in \(R\).  This
vertex shortens no distance between row vertices, and it is within
distance two of every row vertex that has a neighbor, so the answer is
unchanged.  \Cref{thm:ov-lower-bound} now gives
\(Q(\DIAM_{\le2,n})\ge Q(\OV_m)=\Omega(n^{5/4})\).
\end{proof}

\section{An upper bound for Orthogonal Vectors}
\label{sec:ov-upper-bound}
\providecommand{\Coords}[1]{\mathsf{Coords}_{#1}}

In this section, we give an improved quantum algorithm for \(\OV_n\)
using \(\widetilde O(n^{17/12})\) queries.

Two direct approaches give an \(O(n^{3/2})\) query upper bound for
\(\OV_n\). The first searches over all pairs of vectors, testing each
pair by searching for a common one. Each pair test costs
\(O(\sqrt n)\) queries, giving the stated bound. The second searches
over vectors. For each vector, it learns all its one positions and
then searches for another vector that is zero at every one of these
positions. Learning the positions and searching for such a partner
each cost \(O(n)\) queries in the worst case, so the outer search
again gives an \(O(n^{3/2})\) bound.

We first consider a promise version, \(\promOV_n\), in which every
pair has either \(O(\log n)\) or \(\Omega(n/\log n)\) common ones.
This promise retains the difficulties faced by both approaches.
It allows many pairs to have only one or a few common ones, making
coordinate search expensive, while large intersections among other
pairs allow many vectors to have large Hamming weight, making it
expensive to learn all their one positions. These features can occur
simultaneously, even when there is no orthogonal pair. Thus the
promise alone does not improve the \(O(n^{3/2})\) bounds obtained
by these two direct approaches.

Our algorithm combines combinatorial ideas, variable-time quantum
search, and an MNRS quantum walk. For ease of exposition, we first
give the algorithm for \(\promOV_n\) and then extend it to general
inputs.

The error probabilities of the quantum subroutines from
\Cref{sec:quantum-tools} can be made inverse polynomial with an
arbitrarily large fixed exponent at a polylogarithmic overhead.
We assume this error reduction throughout the section.
All samples below are independent uniform draws with replacement,
unless their reuse is stated explicitly.

\subsection{The promise case}
\label{sec:ov-promise}

In this section, we formally define the promise version of the orthogonal vectors problem and give an algorithm for it.

\paragraph{Problem (\(\promOV_n\)).}
The input consists of query access to \(n\) Boolean vectors
\(x_1,\ldots,x_n\in\bits^n\). For fixed positive constants
\(c_{\mathrm{low}}\) and \(c_{\mathrm{high}}\), every pair of distinct indices $i,j \in [n]$ is
promised to have either at most \(c_{\mathrm{low}}\log n\) common ones
or at least \(c_{\mathrm{high}}n/\log n\) common ones. The task is to
decide whether there are distinct indices \(i,j\in[n]\) such that
$\sum_{k=1}^n x_i(k)x_j(k)=0$,
that is, whether two distinct input vectors are orthogonal.

\begin{theorem}[\(\promOV_n\) upper bound]\label{thm:ov-promise-main}
There is a bounded-error quantum query algorithm for \(\promOV_n\) using \(\widetilde O(n^{17/12})\) queries.
\end{theorem}

We now describe the algorithm for $\promOV_n$. Let $s,t\in[n]$ be independent integer parameters that we will eventually choose to be $s=\lceil n^{1/3}\rceil$ and $t=\lceil n^{5/6}\rceil$. We begin by describing a useful preprocessing step.  

\paragraph{Preprocessing and test sets.}

Sample a multiset of \(\Theta(\log^2 n)\) coordinates from $[n]$ independently and uniformly at random, and query them in every vector. Define the candidate graph \(G\) on $[n]$ to have an edge $(i,j)$ between distinct vertices $i, j$ iff none of the sampled coordinates $k$ satisfy that $x_i(k) = x_j(k) = 1$, and write \(N_G(i)\) for the neighbors of the vertex \(i\).

\begin{lemma}[Candidate pairs]\label{lem:ov-promise-candidates}
Any orthogonal pair of distinct vectors $x_i$ and $x_j$ satisfies $\{i,j\}\in E(G)$. Moreover, with probability at
least \(0.99\), every distinct pair of vectors $x_i$ and $x_j$ with $\{i,j\}\in E(G)$ has at
most \(c_{\mathrm{low}}\log n\) common ones.
\end{lemma}
\begin{proof}
An orthogonal pair has no common one at any coordinate, so it always
survives the sample. A pair with at least
\(c_{\mathrm{high}}n/\log n\) common ones is missed by a uniform
coordinate with probability at most \(1-c_{\mathrm{high}}/\log n\).
Independence of the draws therefore bounds its survival probability by
\[
 \left(1-\frac{c_{\mathrm{high}}}{\log n}\right)^{\Theta(\log^2 n)}
 \le \exp\bigl(-\Omega(\log n)\bigr)\le n^{-4},
\]
where the sampling constant is chosen sufficiently large. A union
bound over \(O(n^2)\) pairs eliminates every pair with many
common ones with probability at least \(0.99\). By the promise, each surviving pair then has at most
\(c_{\mathrm{low}}\log n\) common ones.
\end{proof}

We next sample a multiset \(\Pivots\) of \(s\) indices from \([n]\), independently and uniformly at random, and query all coordinates of the corresponding vectors. For each coordinate \(k\), define
\begin{equation}\label{eq:ov-promise-coverage}
 \begin{aligned}
 \Covered{k}&=\bigcup_{p\in\Pivots:\,x_p(k)=1}N_G(p),\\
 \Heavy{k}&=\{i\in[n]:|N_G(i)\setminus\Covered{k}|>t\}.
 \end{aligned}
\end{equation}
Thus $j$ is \textit{covered} (i.e. \(j\in\Covered{k}\)) if some sampled vector \(x_p\) has a one at coordinate \(k\) and \(\{p,j\}\in E(G)\), and \textit{uncovered} otherwise. We call \(x_i\) heavy (i.e. $i \in \Heavy{k}$) at coordinate \(k\) if more than \(t\) of its candidate neighbors lie outside \(\Covered{k}\). These two sets are known after preprocessing, since \(G\) is known and all sampled vectors have been queried.

We use these sets to restrict the coordinates that must be searched when testing whether \(x_i\) and \(x_j\) are orthogonal. Define
\begin{equation}\label{eq:ov-test-sets}
 \begin{aligned}
 \Left{i}&=\{k \in [n]:x_i(k)=1,\ i\in\Heavy{k}\},\\
 \Right{j}&=\{k \in [n]:x_j(k)=1,\ j\in\Covered{k}\},\\
 \Remaining{ij}&=\{k \in [n]:i\notin\Heavy{k},\ j\notin\Covered{k}\},\\
 \Coords{ij}&=\Left{i}\cup\Right{j}\cup\Remaining{ij}.
 \end{aligned}
\end{equation}
Here \(\Coords{ij}\) is the set of coordinates to search, or the search set. The set \(\Remaining{ij}\) is known from preprocessing, whereas the lists \(\Left{i}\) and \(\Right{j}\) will be enumerated only when needed. The following lemma shows that searching \(\Coords{ij}\) suffices and bounds the sizes of its three parts.

\begin{lemma}[Coverage and list bounds]\label{lem:ov-promise-structure}
For any choice of the sets \(\Covered{k}\) and \(\Heavy{k}\), every common one of \(x_i\) and \(x_j\) belongs to \(\Coords{ij}\). With probability at least \(0.98\), the conclusion of \Cref{lem:ov-promise-candidates} and the following bounds hold simultaneously for all \(i,j\in[n]\):
\begin{equation}\label{eq:ov-promise-bounds}
 \sum_i|\Left{i}|\le\frac{100n^3}{st},\qquad
 |\Right{j}|\le c_{\mathrm{low}}s\log n,\qquad
 \sum_{j\in N_G(i)}|\Remaining{ij}|\le nt.
\end{equation}
\end{lemma}
\begin{proof}
Suppose \(x_i(k)=x_j(k)=1\). If \(i\in\Heavy{k}\), then \(k\in\Left{i}\). If \(j\in\Covered{k}\), then \(k\in\Right{j}\). If neither condition holds, then \(k\in\Remaining{ij}\). In every case, \(k\in\Coords{ij}\).

To prove the first bound, condition on the coordinate sample, so that \(G\) is fixed. Fix a coordinate \(k\), and pick the \(s\) pivot draws one at a time. Let \(\mathrm{count}_\ell\) be the number of indices \(i\) for which \(x_i(k)=1\) and more than \(t\) neighbors of \(i\) remain uncovered after the first \(\ell\) draws. The covered set only grows as more pivots are drawn, so \(\mathrm{count}_\ell\) never increases.

Conditioned on the first \(\ell\) draws, the next pivot is one of these \(\mathrm{count}_\ell\) indices with probability \(\mathrm{count}_\ell/n\). If this happens, more than \(t\) previously uncovered indices become covered. The conditional expected number of newly covered indices is therefore at least \(t\,\mathrm{count}_\ell/n\). Since at most \(n\) indices can become covered altogether, taking expectations and summing over the draws gives
\[
 n\ge\frac{t}{n}\sum_{\ell=0}^{s-1}\mathbb E\mathrm{count}_\ell
   \ge\frac{st}{n}\mathbb E\mathrm{count}_s,
\]
implying that $\mathbb E\mathrm{count}_s\le\frac{n^2}{st}$.
The final count is exactly the number of lists \(\Left{i}\) containing \(k\). Summing over the \(n\) coordinates, we obtain
\[
 \mathbb E\sum_i|\Left{i}|\le\frac{n^3}{st}.
\]
Markov's inequality now gives \(\sum_i|\Left{i}|\le100n^3/(st)\) with failure probability at most \(1/100\). 

For the second bound, every \(k\in\Right{j}\) is a common one of \(x_j\) and some sampled vector \(x_p\) with \(\{p,j\}\in E(G)\). On the event in \Cref{lem:ov-promise-candidates}, each sampled vector contributes at most \(c_{\mathrm{low}}\log n\) such coordinates. There are \(s\) pivot draws, so \(|\Right{j}|\le c_{\mathrm{low}}s\log n\).

For the third bound, count the pairs \((j,k)\) with \(j\in N_G(i)\) and \(k\in\Remaining{ij}\) using two different ways:
\[
 \sum_{j\in N_G(i)}|\Remaining{ij}|
 =\sum_{k:\,i\notin\Heavy{k}}|N_G(i)\setminus\Covered{k}|\le nt,
\]
where the inequality follows by the definition of \(\Heavy{k}\). Combining the sampling guarantee from \Cref{lem:ov-promise-candidates} with the first bound's failure probability proves that all the conclusions hold with probability at least \(0.98\).
\end{proof}

We next group the vectors according to estimates of their left list sizes. Independently sample a multiset \(K\) of \(\Theta((n/t)\log n)\) coordinates and query them in every vector. For each \(i\in[n]\), define \(\Size{i}\) to be the least power of two no smaller than
\begin{equation}\label{eq:ov-size-definition}
 \max\left\{2t,\frac{2n}{|K|}
                \sum_{k\in K}\ind{k\in\Left{i}}\right\}.
\end{equation}
Membership in \(\Left{i}\) at a sampled coordinate is determined by the queried bit $x_i(k)$ and the known set \(\Heavy{k}\). Thus \(\Size{i}\) is known for all $i$ after these queries. For each power of two \(a\) with \(2t\le a\le4n\), define
\begin{equation}\label{eq:ov-classes}
 V_a=\{i:\Size{i}=a\},\qquad n_a=|V_a|.
\end{equation}
The set \(V_a\) contains the indices of vectors with size label \(a\).
We will use the same definitions in the general case. The next lemma shows that each size label bounds the corresponding list size, while the sum of the size labels remains small.

\begin{lemma}[Size labels]\label{lem:ov-size}
Conditioned on \(\sum_i|\Left{i}|\le100n^3/(st)\), with probability at least $1-2n^{-3}$ the labels in \eqref{eq:ov-size-definition} satisfy
\begin{equation}\label{eq:ov-size-bounds}
 |\Left{i}|\le\Size{i}\le8(|\Left{i}|+t),\qquad
 a n_a\le\sum_i\Size{i}\le8nt+\frac{800n^3}{st}.
\end{equation}
Moreover, there are \(O(\log n)\) nonempty size classes.
\end{lemma}
\begin{proof}
Let \(m=|K|\). For a fixed vector \(x_i\), let \(\ell\) count the sampled coordinates in \(\Left{i}\), with multiplicity. Then \(\ell\) is binomial with mean \(m|\Left{i}|/n\), so \(n\ell/m\) estimates \(|\Left{i}|\). Chernoff bounds give
\[
 \Pr\!\left[\frac{n\ell}{m}>2(|\Left{i}|+t)\right]\le e^{-mt/(3n)},
 \qquad
 \Pr\!\left[\frac{n\ell}{m}<\frac {|\Left{i}|} 2\right]
 \le e^{-m |\Left{i}|/(8n)},
\]
which is $\le e^{-mt/(4n)}$ if $|\Left{i}|>2t$.
For the first inequality, the threshold for \(\ell\) is at least twice its mean and exceeds that mean by at least \(mt/n\). The second inequality is the usual lower tail bound. Since \(m=\Theta((n/t)\log n)\), choosing the sampling constant sufficiently large makes each failure probability at most \(n^{-4}\). A union bound makes both estimates hold for all vectors with probability at least \(1-2n^{-3}\).
On this event, if \(|\Left{i}|\le2t\), the definition of \(\Size{i}\) implies \(\Size{i}\ge|\Left{i}|\). If \(|\Left{i}|>2t\), the lower estimate gives \(2n\ell/m\ge |\Left{i}|\), so the same conclusion holds. 
Rounding up to a power of two increases a positive number by at most a factor of two. 
The upper estimate therefore gives
\[
 \Size{i}\le2\max\{2t,2n\ell/m\}\le8(|\Left{i}|+t).
\]
Summing over all vectors and using the assumed list bound, we obtain
\[
 \sum_i\Size{i}\le8\left(\sum_i|\Left{i}|+nt\right)
 \le8nt+\frac{800n^3}{st}.
\]
Since every index in \(V_a\) has size label \(a\), we also have \(a n_a=\sum_{i\in V_a}\Size{i}\le8nt+800n^3/(st)\). Finally, all size labels are powers of two between \(2t\) and \(4n\) so there are \(O(\log n)\) nonempty size classes.
\end{proof}

\paragraph{The checker.}

Let $r \in [n]$ be an integer parameter that we will eventually choose to be \(r=\lfloor\sqrt{n_a}\rfloor\). We now describe a subroutine that, given a set of vectors of size $r$ and information about the left lists of elements in this set, checks whether any of them has an orthogonal partner anywhere in the input. Throughout the rest of this subsection, condition on the events in \Cref{lem:ov-promise-structure,lem:ov-size}.

Fix a nonempty set \(I\subseteq V_a\), let \(r=|I|\), and suppose that the lists \((\Left{i})_{i\in I}\) have been stored. For a fixed \(j\in[n]\), first enumerate \(\Right{j}\). We can then test any candidate \(i\in I\cap N_G(j)\) by searching \(\Coords{ij}\) for a common one of \(x_i\) and \(x_j\). By \Cref{lem:ov-promise-structure}, this test accepts exactly when the two vectors are orthogonal.

We apply the known-cost consequence of \Cref{lem:unknown-cost-search} first over \(i\in I\cap N_G(j)\), and then over \(j\in[n]\). To bound the query costs, we will define
\begin{equation}\label{eq:ov-promise-budget}
 \Budget{j}=r(a+s)+\sum_{i\in I\cap N_G(j)}|\Remaining{ij}|.
\end{equation}
These quantities are known from preprocessing and \(I\). 
Algorithm~\ref{alg:ov-promise-check} provides the checker. Each search returns \textsc{Yes} if it finds an accepting test and \textsc{No} otherwise.

\begin{algorithm}[H]
\caption{Promise checker}\label{alg:ov-promise-check}
\small
\begin{algorithmic}[1]
\Require \(I\subseteq V_a\), \(r=|I|\), and \((\Left{i})_{i\in I}\).
\Function{PromiseHasPartner}{$j$}
  \State Enumerate \(\Right{j}\) using \Cref{lem:implicit-set-access}, part~3.
  \State For \(i\in I\cap N_G(j)\), search \(\Coords{ij}\) for \(x_i(k)x_j(k)=1\) and accept if no such \(k\) is found.
  \State \Return search over \(i\in I\cap N_G(j)\) using \Cref{lem:unknown-cost-search}.
\EndFunction
\Function{PromiseCheck}{$I,a$}
  \State \Return search over \(j\in[n]\) using \Call{PromiseHasPartner}{$j$} and \Cref{lem:unknown-cost-search}.
\EndFunction
\end{algorithmic}
\end{algorithm}

\begin{lemma}[Promise checking cost]\label{lem:ov-promise-check}
Given \(I\subseteq V_a\) and the lists \((\Left{i})_{i\in I}\), Algorithm~\ref{alg:ov-promise-check} decides whether some \(x_i\), with \(i\in I\), has an orthogonal partner. Its query cost is
\begin{equation}\label{eq:ov-checking}
 C=\Ot\!\left(n\sqrt s+\sqrt{nra}\right).
\end{equation}
\end{lemma}
\begin{proof}
Fix \(j\in[n]\). Testing whether \(k\in\Right{j}\) requires at most one query to \(x_j(k)\), since \(\Covered{k}\) is known. By \Cref{lem:implicit-set-access} and the bound on \(|\Right{j}|\), enumerating this list costs \(\Ot(\sqrt{ns})\) queries.

Once \(\Right{j}\) is stored, the search set \(\Coords{ij}\) is known for every \(i\in I\cap N_G(j)\). Grover search over this set uses
\[
 \Ot\!\left(\sqrt{a+s\log n+|\Remaining{ij}|+1}\right)
 =\Ot\!\left(\sqrt{a+s+|\Remaining{ij}|}\right)
\]
queries, where we use \(|\Left{i}|\le a\). 

By \Cref{lem:unknown-cost-search}, searching over the candidate indices \(i\) therefore costs \(\Ot(\sqrt{\Budget{j}})\). 

It remains to search over \(j\). The total budget satisfies
\[
 \sum_j\Budget{j}
 =nr(a+s)+\sum_{i\in I}\sum_{j\in N_G(i)}|\Remaining{ij}|
 \le nr(a+s)+rnt=O(nr(a+s)),
\]
where we use \Cref{lem:ov-promise-structure} and \(a\ge2t\). A second application of \Cref{lem:unknown-cost-search} gives
\[
 C=\Ot\!\left(\sqrt{\sum_j(ns+\Budget{j})}\right)
   =\Ot\!\left(n\sqrt s+\sqrt{nra}\right),
\]
using \(r\le n\).
Every orthogonal pair is an edge of \(G\), and the search set contains every common one. Thus the checker accepts exactly when a vector indexed by \(I\) has an orthogonal partner, up to the error probability of the quantum searches.
\end{proof}

\paragraph{The quantum walk.}

We now use this checker in a quantum walk over each nonempty class \(V_a\). For each such class, let \(r=\lfloor\sqrt{n_a}\rfloor\). The states are the \(r\)-element subsets \(I\in\binom{V_a}{r}\), and each state stores the lists \((\Left{i})_{i\in I}\). A state is marked if some \(x_i\), with \(i\in I\), has an orthogonal partner anywhere in the input.

We use the Johnson walk from \Cref{sec:quantum-tools}. Each step replaces one index in \(I\) and updates its stored list. Let \(S\) be the query cost of preparing a state with its lists, and let \(U\) be the query cost of one replacement. The checker has cost \(C\) from \Cref{lem:ov-promise-check}. The full algorithm for $\promOV_n$ is given below.

\begin{algorithm}[H]
\caption{Promise algorithm}\label{alg:ov-promise}
\small
\begin{algorithmic}[1]
\Require \(x_1,\ldots,x_n\) satisfying the promise of \(\promOV\).
\State \(s\gets\lceil n^{1/3}\rceil\), \(t\gets\lceil n^{5/6}\rceil\).
\State Query \(\Theta(\log^2 n)\) sampled coordinates in all vectors; construct \(G\).
\State \(\Pivots\gets s\) samples from \([n]\); query all pivot entries.
\State \(\Covered{k},\Heavy{k}\gets\eqref{eq:ov-promise-coverage}\) for \(k\in[n]\).
\State \(K\gets\Theta((n/t)\log n)\) independent samples from \([n]\); query all vectors on \(K\).
\State \((\Size{i})_{i\in[n]},(V_a)_a\gets\eqref{eq:ov-size-definition},\eqref{eq:ov-classes}\).
\For{\(a\) with \(n_a>0\)}
  \State \(r\gets\lfloor\sqrt{n_a}\rfloor\).
  \State Run MNRS on \(I\in\binom{V_a}{r}\), storing \((\Left{i})_{i\in I}\),
  \Statex \hspace{3em} with checker \Call{PromiseCheck}{$I,a$}.
  \If{the search accepts} \State \Return \textsc{Yes}. \EndIf
\EndFor
\State \Return \textsc{No}.
\end{algorithmic}
\end{algorithm}

\begin{proof}[Proof of \Cref{thm:ov-promise-main}]
We analyze Algorithm~\ref{alg:ov-promise}, keeping \(s,t\) as independent parameters until the final bound. Membership in \(\Left{i}\) requires at most one query, since \(\Heavy{k}\) is known for every \(k\). By \Cref{lem:implicit-set-access}, enumerating a left list from class \(V_a\) costs \(\Ot(\sqrt{na})\) queries. Preparing \(r\) such lists and replacing one of them therefore have respective costs
\begin{equation}\label{eq:ov-promise-su}
 S=\Ot(r\sqrt{na}),\qquad U=\Ot(\sqrt{na}).
\end{equation}
Suppose a vector with an orthogonal partner has its index in \(V_a\). At least an \(r/n_a\) fraction of the walk states contain that index and are therefore marked. The Johnson walk has spectral gap \(\Theta(1/r)\), so \Cref{thm:mnrs} bounds the cost of searching this class by\footnote{For \(n_a\le2\), prepare the lists and run the checker for each walk state, at total cost \(O(S+C)\), within the same bound.}
\[
 \Ot\!\left(S+\sqrt{n_a}\,U+\sqrt{n_a/r}\,C\right).
\]
Substituting the setup, update, and checking costs, and using \(r=\lfloor\sqrt{n_a}\rfloor\), gives
\begin{equation}\label{eq:ov-class-cost}
 \Ot\!\left(\sqrt{na n_a}+n\sqrt s\,n_a^{1/4}\right)
 \le\Ot\!\left(n\sqrt t+\frac{n^2}{\sqrt{st}}+n^{5/4}\sqrt s\right).
\end{equation}
Here we use \(a n_a=O(nt+n^3/(st))\) from \Cref{lem:ov-size} and \(n_a\le n\).

The preprocessing uses \(\Ot(ns+n^2/t)\) queries. There are \(O(\log n)\) nonempty size classes, so the total query cost is
\begin{equation}\label{eq:ov-total}
 \Ot\!\left(ns+\frac{n^2}{t}+n\sqrt t
       +\frac{n^2}{\sqrt{st}}+n^{5/4}\sqrt s\right).
\end{equation}
Choosing \(s=\lceil n^{1/3}\rceil\) and \(t=\lceil n^{5/6}\rceil\) makes this \(\Ot(n^{17/12})\). We truncate the algorithm at a sufficiently large constant times the total query bound above, counting truncation as failure.

Finally, every orthogonal pair survives preprocessing, and each of its indices belongs to a size class. A walk state containing either index is marked, and \Cref{lem:ov-promise-check} gives a correct checker. If no orthogonal pair exists, no walk state is marked. The preprocessing and size label events hold with probability at least \(0.97\); reducing the search errors makes the overall success probability at least \(2/3\).
\end{proof}

\subsection{The general case}
\label{sec:ov-general}

In this section, we give an algorithm for the orthogonal vectors problem without the promise on pairwise intersections.


\begin{theorem}[\(\OV_n\) upper bound]\label{thm:ov-upper-bound}
There is a bounded-error quantum query algorithm for \(\OV_n\) using \(\widetilde O(n^{17/12})\) queries.
\end{theorem}

Consequently, from \Cref{thm:bmpv-final}, we get an improved upper bound for $\BMPV$.

\begin{corollary}[\(\BMPV_n\) upper bound]\label{thm:bmpv-upper-bound}
There is a bounded-error quantum query algorithm for \(\BMPV_n\) using \(\widetilde O(n^{17/12})\) queries.
\end{corollary}

We now describe the algorithm for $\OV_n$. Let $s,t\in[n]$ be independent integer parameters that we will eventually choose to be $s=\lceil n^{1/3}\rceil$ and $t=\lceil n^{5/6}\rceil$. We use the same test sets, size classes, and quantum walk as in the promise case. The main difference is that candidate neighborhoods now depend on the coordinate: a pair remains a candidate until its first common one. We will also need subroutines to access the test sets, since these are no longer known from preprocessing.

\paragraph{Preprocessing and test sets.}

For distinct indices \(i,j\in[n]\), let \(\First{i}{j}\) be the least coordinate \(k\) with \(x_i(k)x_j(k)=1\), or \(n+1\) if the two vectors are orthogonal. Define $\First{i}{i}=0$ and
\begin{equation}\label{eq:ov-first}
 N_k(i)=\{j\in[n]:\First{i}{j}\ge k\}.
\end{equation}
Thus \(j\in N_k(i)\) iff \(i\ne j\) and the two vectors have no common one before coordinate \(k\). These sets take the place of the fixed neighborhoods \(N_G(i)\) in the promise case. We use them implicitly, without computing all pairwise first intersections.

Independently sample multisets \(\Pivots\) and \(\DegSamp\) of indices from \([n]\), uniformly at random, and query all coordinates of the corresponding vectors. The first sample is used for defining the covered sets, while the second is used for estimating the number of uncovered candidate neighbors. As in the promise case, \(s\) is the number of pivot draws. We choose
\begin{equation}\label{eq:ov-pools}
 |\Pivots|=s,\quad |\DegSamp|=\Theta\left(\frac nt\log n\right),\quad
 |\Pivots|+|\DegSamp|=\Ot(s+n/t).
\end{equation}
For each coordinate \(k\), define
\begin{equation}\label{eq:ov-general-coverage}
 \begin{aligned}
 \Covered{k}&=\bigcup_{p\in\Pivots:\,x_p(k)=1}N_k(p),\\
 \Heavy{k}&=\left\{i\in[n]:
 \frac{n}{|\DegSamp|}\sum_{p\in\DegSamp}
 \ind{\First{i}{p}\ge k}\ind{p\notin\Covered{k}}>2t\right\}.
 \end{aligned}
\end{equation}
As before, \(j\) is covered at coordinate \(k\) if some sampled vector has a one there and has \(j\) as a candidate neighbor. We call \(x_i\) heavy at \(k\) if the estimated number of its uncovered candidate neighbors exceeds \(2t\). The independent degree sample ensures that this estimate is accurate even though the covered sets depend on \(\Pivots\).

Using the covered and heavy sets in \eqref{eq:ov-general-coverage}, define
\begin{equation}\label{eq:ov-general-test-sets}
 \begin{aligned}
 \Left{i}&=\{k\in[n]:x_i(k)=1,\ i\in\Heavy{k}\},\\
 \Right{j}&=\{k\in[n]:x_j(k)=1,\ j\in\Covered{k}\},\\
 \Remaining{ij}&=\{k\in[n]:i\notin\Heavy{k},\ j\notin\Covered{k}\},\\
 \Coords{ij}&=\Left{i}\cup\Right{j}\cup\Remaining{ij}.
 \end{aligned}
\end{equation}
Unlike in the promise case, these sets are not known after preprocessing. We first prove their size bounds, and then describe how to access them.

\begin{lemma}[General preprocessing bounds]\label{lem:ov-general-structure}
For any choice of the sets \(\Covered{k}\) and \(\Heavy{k}\), every common one of \(x_i\) and \(x_j\) belongs to \(\Coords{ij}\). With probability at least \(0.98\), the following bounds hold simultaneously for all indices and coordinates:
\begin{equation}\label{eq:ov-degree-bounds}
 \begin{aligned}
 i\in\Heavy{k}&\ \Longrightarrow\ |N_k(i)\setminus\Covered{k}|>t,\\
 i\notin\Heavy{k}&\ \Longrightarrow\ |N_k(i)\setminus\Covered{k}|\le4t,
 \end{aligned}
\end{equation}
and
\begin{equation}\label{eq:ov-general-bounds}
 \sum_i|\Left{i}|\le\frac{100n^3}{st},\qquad |\Right{j}|\le s,\qquad
 \sum_{j\ne i}\bigl|\{k\in\Remaining{ij}:k\le\First{i}{j}\}\bigr|\le4nt.
\end{equation}
\end{lemma}
\begin{proof}
The coverage statement follows from the same three cases as in \Cref{lem:ov-promise-structure}.

Condition on \(\Pivots\), so that every covered set is fixed. For fixed \(i,k\), the number of occurrences in \(\DegSamp\) belonging to \(N_k(i)\setminus\Covered{k}\) is binomial, with mean
\[
 \frac{|\DegSamp|}{n}|N_k(i)\setminus\Covered{k}|.
\]
If the true number of uncovered neighbors is at most \(t\), a Chernoff bound shows that its estimate exceeds \(2t\) with probability at most \(\exp(-\Omega(|\DegSamp|t/n))\). If the true number exceeds \(4t\), the same bound shows that its estimate is at most \(2t\) with at most this probability. Choosing the sampling constant sufficiently large and taking a union bound over all \(n^2\) choices of \(i,k\) proves \eqref{eq:ov-degree-bounds} with failure probability at most \(1/100\).

For the left list bound, fix a coordinate \(k\) and pick the pivot draws one at a time. The neighborhoods \(N_k(i)\) are fixed, so the counting argument in \Cref{lem:ov-promise-structure} applies without change: sampling an index \(i\) with \(x_i(k)=1\) and more than \(t\) uncovered neighbors covers more than \(t\) new indices. Consequently,
\[
 \mathbb E\sum_k
 \bigl|\{i:x_i(k)=1,\ |N_k(i)\setminus\Covered{k}|>t\}\bigr|
 \le\frac{n^3}{|\Pivots|t}=\frac{n^3}{st}.
\]
By Markov's inequality, this sum is at most \(100n^3/(st)\) with failure probability at most \(1/100\). On the degree estimation event, every entry in a left list is counted in this sum, proving \(\sum_i|\Left{i}|\le100n^3/(st)\). 

For the right lists, we have the exact identity
\begin{equation}\label{eq:ov-right-first}
 \Right{j}=\{\First{j}{p}:p\in\Pivots\}\cap[n].
\end{equation}
Indeed, if \(k\in\Right{j}\), then \(x_j(k)=x_p(k)=1\) for some pivot \(p\) with \(\First{j}{p}\ge k\). These conditions imply \(\First{j}{p}=k\). Conversely, a first common one with a pivot belongs to \(\Right{j}\). Each pivot contributes at most one coordinate, so \(|\Right{j}|\le|\Pivots|=s\).

Finally, since \(j\in N_k(i)\) iff \(i\ne j\) and \(k\le\First{i}{j}\), counting the same pairs \((j,k)\) in two orders gives
\[
 \sum_{j\ne i}\bigl|\{k\in\Remaining{ij}:k\le\First{i}{j}\}\bigr|
 =\sum_{k:\,i\notin\Heavy{k}}|N_k(i)\setminus\Covered{k}|\le4nt,
\]
where the inequality follows by summing the second bound in \eqref{eq:ov-degree-bounds} over coordinates. Combining the two failure probabilities proves the lemma.
\end{proof}

We next group the vectors according to estimates of their left list sizes. Independently sample a multiset \(K\) of \(\Theta((n/t)\log n)\) coordinates. For each \(i\in[n]\), define \(\Size{i}\) to be the least power of two no smaller than
\begin{equation}\label{eq:ov-general-size-definition}
 \max\left\{2t,\frac{2n}{|K|}
              \sum_{k\in K}\ind{k\in\Left{i}}\right\}.
\end{equation}
For each power of two \(a\) with \(2t\le a\le4n\), define
\begin{equation}\label{eq:ov-general-classes}
 V_a=\{i\in[n]:\Size{i}=a\},\qquad n_a=|V_a|.
\end{equation}
Thus \(V_a\) contains the indices of vectors with size label \(a\), and \(n_a\) is their number. The sample \(K\) is independent of the covered and heavy sets, so \Cref{lem:ov-size} applies to these labels as well. We evaluate a label only when needed, rather than querying \(K\) in every vector during preprocessing. Throughout the rest of this subsection, condition on the events in \Cref{lem:ov-general-structure,lem:ov-size}.

\paragraph{Access to the test sets.}

To determine whether \(i\) belongs to \(\Covered{k}\) or \(\Heavy{k}\), it suffices to know its first intersection with every sampled vector. Define the record of \(x_i\) by
\begin{equation}\label{eq:ov-record}
 \Record{i}=(\First{i}{p})_{p\in\Pivots\sqcup\DegSamp},
\end{equation}
where \(\sqcup\) concatenates the two multisets, keeping one entry for each occurrence. The maintained data for \(x_i\) consist of \(\Record{i}\) together with \(\Left{i}\).

We load a record by finding successive common ones between \(x_i\) and sampled vectors whose first intersection has not yet been found. Each discovered coordinate determines the first intersection for every such sampled vector having a one there. Algorithm~\ref{alg:ov-record} gives the procedure. The set \(\mathit{unresolved}\) contains these sample occurrences, and \(\mathit{start}\) is the first coordinate still to be searched. Each search for the least coordinate uses \Cref{lem:successive-discoveries}.

\begin{algorithm}[H]
\caption{Load a record}\label{alg:ov-record}
\small
\begin{algorithmic}[1]
\Require \(i\in[n]\) and the fully read samples \(\Pivots,\DegSamp\).
\State \((\Record{i})_p\gets(n+1)\cdot\ind{p\ne i}\) for \(p\in\Pivots\sqcup\DegSamp\).
\State \(\mathit{unresolved}\gets\{p\in\Pivots\sqcup\DegSamp:p\ne i\}\); \(\mathit{start}\gets1\).
\While{\(\mathit{unresolved}\ne\varnothing\)}
  \State Find the least \(k\ge\mathit{start}\) with \(x_i(k)=1\) and \(\bigvee_{p\in\mathit{unresolved}}x_p(k)=1\).
  \If{no such \(k\) exists} \State \Return \(\Record{i}\). \EndIf
  \State \((\Record{i})_p\gets k\) for \(p\in\mathit{unresolved}\) with \(x_p(k)=1\).
  \State \(\mathit{unresolved}\gets\{p\in\mathit{unresolved}:x_p(k)=0\}\); \(\mathit{start}\gets k+1\).
\EndWhile
\State \Return \(\Record{i}\).
\end{algorithmic}
\end{algorithm}

\begin{lemma}[Record loading]\label{lem:ov-record}
Algorithm~\ref{alg:ov-record} loads \(\Record{i}\) in \(\Ot(\sqrt{n(s+n/t)})\) queries. The record determines membership of \(i\) in every \(\Covered{k}\) and \(\Heavy{k}\), and determines \(\Right{i}\). For \(i\in V_a\), loading its maintained data costs \(\Ot(\sqrt{na}+\sqrt{n(s+n/t)})\) queries.
\end{lemma}
\begin{proof}
The search predicate requires at most one query to \(x_i(k)\), since all sampled vectors have been queried. Each successful search determines at least one previously unknown first intersection, so there are at most \(s+|\DegSamp|\) iterations. Removing resolved occurrences can only remove possible common ones, so no earlier coordinate needs to be searched again. By \Cref{lem:successive-discoveries} and \eqref{eq:ov-pools}, the total cost is \(\Ot(\sqrt{n(s+n/t)})\). Each stored value is the earliest common one for its occurrence, and entries left at \(n+1\) correspond to orthogonal vectors.

Given \(\Record{i}\), membership in \(\Covered{k}\) is determined by the comparisons \(\First{i}{p}\ge k\) and the known bits \(x_p(k)\) for \(p\in\Pivots\). The comparisons needed for membership in \(\Heavy{k}\) are also stored. Moreover, whether a degree sample vector belongs to \(\Covered{k}\) is known from the fully queried samples. Thus \eqref{eq:ov-general-coverage} can be evaluated without further queries. The right list is determined by \eqref{eq:ov-right-first}.

Membership in \(\Left{i}\) then requires at most one query to \(x_i(k)\). Since \(|\Left{i}|\le a\), \Cref{lem:implicit-set-access}, part~3 enumerates this list in \(\Ot(\sqrt{na})\) queries. Including the record, the cost is
\begin{equation}\label{eq:ov-left-load}
 \Ot\!\left(\sqrt{n(s+n/t)}+\sqrt{na}\right).
\end{equation}
\end{proof}

\paragraph{The checker.}

We now describe a subroutine that, given a set \(I\subseteq V_a\) and its maintained data, checks whether some \(x_i\), with \(i\in I\), has an orthogonal partner anywhere in the input. Fix a nonempty such set and let \(r=|I|\). Once \(\Record{j}\) is loaded, \(\Coords{ij}\) is known for every \(i\in I\setminus\{j\}\): the left list is stored, while the right list and remaining set are determined by the records.

In the promise case, we bounded the cost of searching the entire search set. Here, we search coordinates in increasing order and stop once a common one is found. To bound this cost, define
\begin{equation}\label{eq:ov-work}
 \begin{aligned}
 \Work{ij}&=\bigl|\{k\in\Coords{ij}\cup\{n+1\}:k\le\First{i}{j}\}\bigr|
       &&(i\ne j),\\
 \Work{j}&=\sum_{i\in I\setminus\{j\}}\Work{ij}.
 \end{aligned}
\end{equation}
For an intersecting pair, \(\Work{ij}\) is the position of its first common one in the increasingly ordered search set. For an orthogonal pair, it is \(|\Coords{ij}|+1\). The extra coordinate \(n+1\) accounts for finishing a search with no common one.

\begin{lemma}[Total work]\label{lem:ov-work}
The quantities in \eqref{eq:ov-work} satisfy
\begin{equation}\label{eq:ov-work-bounds}
 \sum_{j\ne i}\Work{ij}=O(n(a+s))\quad(i\in I),\qquad
 \sum_j\Work{j}=O(nr(a+s)).
\end{equation}
\end{lemma}
\begin{proof}
For each pair, the left list, right list, and extra coordinate contribute at most \(|\Left{i}|+s+1\). A remaining coordinate contributes only while \(k\le\First{i}{j}\). Then,
\begin{align*}
 \sum_{j\ne i}\Work{ij}
 &\le n(1+|\Left{i}|+s)
   +\sum_{j\ne i}\bigl|\{k\in\Remaining{ij}:k\le\First{i}{j}\}\bigr|\\
 &\le n(1+|\Left{i}|+s+4t)=O(n(a+s)).
\end{align*}
The inequality follows from \Cref{lem:ov-general-structure}; the final equality uses \(|\Left{i}|\le a\) and \(a\ge2t\). Summing over the \(r\) indices in \(I\) proves the second bound.
\end{proof}

We test a pair by searching successively doubled prefixes of its search set, in increasing coordinate order. Algorithm~\ref{alg:ov-subroutines} gives this procedure. We apply \Cref{lem:unknown-cost-search} to truncated runs of \textsc{PairTest}, returning \textsc{Unfinished} whenever a run is stopped before completion.

\begin{algorithm}[H]
\caption{Pair testing}\label{alg:ov-subroutines}
\small
\begin{algorithmic}[1]
\Require \(i\in I\), \(j\ne i\), and the stored \((\Record{i},\Left{i})\) and \(\Record{j}\).
\Function{PairTest}{$i,j$}
  \State Use the stored \((\Record{i},\Left{i})\) and \(\Record{j}\).
  \For{prefixes of \(\Coords{ij}\) of lengths \(1,2,4,\ldots\), ending with the full set}
    \State Search the prefix for a common one.
    \If{a common one is found} \State \Return \textsc{Nonorthogonal}. \EndIf
  \EndFor
  \State \Return \textsc{Orthogonal}.
\EndFunction
\end{algorithmic}
\end{algorithm}

We can now describe the checker. To search for an orthogonal pair involving a fixed \(j\), apply \Cref{lem:unknown-cost-search} over \(i\in I\setminus\{j\}\), using \textsc{PairTest}. This does not require knowing \(\Work{j}\). To search over \(j\), apply \Cref{lem:variable-stopping-search} to the resulting tests. Algorithm~\ref{alg:ov-general-check} gives the checker; each search accepts if it finds an orthogonal pair.

\begin{algorithm}[H]
\caption{General checker}\label{alg:ov-general-check}
\small
\begin{algorithmic}[1]
\Require \(I\subseteq V_a\), \(r=|I|\), and \((\Record{i},\Left{i})_{i\in I}\).
\Function{GeneralHasPartner}{$j$}
  \State Load \(\Record{j}\) by Algorithm~\ref{alg:ov-record}.
  \State Search over \(i\in I\setminus\{j\}\) using \Cref{lem:unknown-cost-search}.
  \State Use \Call{PairTest}{$i,j$} as the test, accepting \textsc{Orthogonal}.
  \State \Return the search result.
\EndFunction
\Function{GeneralCheck}{$I,a$}
  \State \Return search over \(j\in[n]\) using \Call{GeneralHasPartner}{$j$} and \Cref{lem:variable-stopping-search}.
\EndFunction
\end{algorithmic}
\end{algorithm}

\begin{lemma}[General checking cost]\label{lem:ov-general-check}
Given \(I\subseteq V_a\) and its maintained data, Algorithm~\ref{alg:ov-general-check} decides whether some \(x_i\), with \(i\in I\), has an orthogonal partner. Its query cost is
\begin{equation}\label{eq:ov-general-checking}
 C=\Ot\!\left(n\sqrt{s+n/t}+\sqrt{nra}\right).
\end{equation}
\end{lemma}
\begin{proof}
Fix \(j\in[n]\). Because the prefix lengths grow geometrically, \textsc{PairTest} uses \(\Ot(\sqrt{\Work{ij}})\) queries. By \Cref{lem:unknown-cost-search}, searching over \(i\) therefore uses \(\Ot(\sqrt{\Work{j}})\) queries with high probability.

Loading \(\Record{j}\) costs \(\Ot(\sqrt{n(s+n/t)})\) queries by \Cref{lem:ov-record}. Thus, for \textsc{GeneralHasPartner}, the completion bound in \Cref{lem:variable-stopping-search} can be taken as \(q_j=\Ot(\sqrt{n(s+n/t)+\Work{j}})\). The record is loaded once and reused throughout the search over \(i\).

It remains to search over \(j\). By \Cref{lem:ov-work},
\begin{align*}
 \sqrt{\sum_j q_j^2}
  &=\Ot\!\left(\sqrt{\sum_j\bigl(n(s+n/t)+\Work{j}\bigr)}\right)\\
  &=\Ot\!\left(n\sqrt{s+n/t}+\sqrt{nr(a+s)}\right)
   =\Ot\!\left(n\sqrt{s+n/t}+\sqrt{nra}\right),
\end{align*}
using \(r\le n\). Taking \(Q\) to be the final bound and applying \Cref{lem:variable-stopping-search} gives the stated checking cost, without estimating the individual quantities \(\Work{j}\).

The search sets contain every common one, so the checker accepts exactly when a vector indexed by \(I\) has an orthogonal partner, up to the error probability of the quantum searches.
\end{proof}

\begin{remark}[Stopping times in \Cref{alg:ov-general-check}]
The proof of \Cref{lem:ov-general-check} applies \Cref{lem:variable-stopping-search} with $\textsc{GeneralHasPartner}(j)$ as the test for candidate~$j$. Its stopping time depends on intermediate measurement outcomes in \Cref{alg:ov-record} and in the search over~$i$. Deferring these measurements gives an algorithm that sets a stopping flag at fixed checkpoints and has the same stopping time distribution, which is the form required by variable-time amplitude amplification~\cite{Ambainis2012VariableTimeAmplification}. See also~\cite{Jeffery2022SubroutineComposition,BelovsJefferyYolcu2024} for general composition of quantum subroutines with variable stopping times.
\end{remark}

\paragraph{The quantum walk.}

We now use this checker in a quantum walk over each nonempty class \(V_a\). As in the promise case, the states are the \(r\)-element subsets \(I\in\binom{V_a}{r}\), and a state is marked if some \(x_i\), with \(i\in I\), has an orthogonal partner anywhere in the input. Each state now stores the maintained data \((\Record{i},\Left{i})_{i\in I}\).

The size classes are no longer known explicitly, so preparing and updating a state also requires searching for members of its class. The next lemma bounds this additional cost. As before, \(S\) is the query cost of preparing a state with its data, and \(U\) is the cost of one replacement in the Johnson walk.

\begin{lemma}[Access to implicit classes]\label{lem:ov-class-access}
Class membership costs \(\Ot(\sqrt{n(s+n/t)})\) queries, and estimating \(n_a\) within a factor of two costs \(\Ot(n\sqrt{s+n/t})\) queries. For a nonempty class and \(1\le r\le n_a/2\), the setup and update costs are
\begin{equation}\label{eq:ov-general-su}
 S=\Ot\!\left(r\sqrt{na}+rn\sqrt{(s+n/t)/n_a}\right),\qquad
 U=\Ot\!\left(\sqrt{na}+n\sqrt{(s+n/t)/n_a}\right).
\end{equation}
\end{lemma}
\begin{proof}
To compute \(\Size{i}\), load \(\Record{i}\) and query \(x_i\) at the coordinates in \(K\). The cost is
\[
 \Ot\!\left(\sqrt{n(s+n/t)}+n/t\right)=\Ot(\sqrt{n(s+n/t)}).
\]
Testing \(\Size{i}=a\) gives class membership. By \Cref{lem:implicit-set-access}, parts~1 and~2, estimating the class size costs \(\Ot(n\sqrt{s+n/t})\), and preparing a uniform superposition over elements of \(V_a\) costs \(\Ot(n\sqrt{(s+n/t)/n_a})\).

Preparing a state requires \(r\) selections without replacement and loading their maintained data. Since \(r\le n_a/2\), excluding previously selected indices changes each selection cost by at most a constant factor. An update removes one stored index and its maintained data, selects a new index outside the subset, and loads its maintained data. Combining the selection cost with \Cref{lem:ov-record} gives the stated setup and update costs; the selection cost absorbs the record loading term since \(n_a\le n\).
\end{proof}

For each nonempty class, we obtain an estimate \(\widehat n_a\in[n_a/2,2n_a]\) and choose \(r=\max\{1,\lfloor\sqrt{\widehat n_a/2}\rfloor\}\). This ensures \(r=\Theta(\sqrt{n_a})\). The full algorithm for $\OV_n$ is given below.

\begin{algorithm}[H]
\caption{General algorithm}\label{alg:ov-general}
\small
\begin{algorithmic}[1]
\Require \(x_1,\ldots,x_n\in\bits^n\).
\State \(s\gets\lceil n^{1/3}\rceil\), \(t\gets\lceil n^{5/6}\rceil\).
\State Sample \(\Pivots,\DegSamp\) as in \eqref{eq:ov-pools}; query all their entries.
\State \(K\gets\Theta((n/t)\log n)\) independent samples from \([n]\).
\State Access \(\Covered{k},\Heavy{k},\Size{i},V_a\) by \eqref{eq:ov-general-coverage}, \eqref{eq:ov-general-size-definition}, \eqref{eq:ov-general-classes}.
\For{powers of two \(a\) with \(2t\le a\le4n\)}
  \State Estimate \(n_a\) by \(\widehat n_a\in[n_a/2,2n_a]\) using \Cref{lem:implicit-set-access}; skip empty classes.
  \State \(r\gets\max\{1,\lfloor\sqrt{\widehat n_a/2}\rfloor\}\).
  \State Store \((\Record{i},\Left{i})_{i\in I}\) for each state \(I\in\binom{V_a}{r}\).
  \State Run MNRS with checker \Call{GeneralCheck}{$I,a$}.
  \If{the search accepts} \State \Return \textsc{Yes}. \EndIf
\EndFor
\State \Return \textsc{No}.
\end{algorithmic}
\end{algorithm}

\begin{proof}[Proof of \Cref{thm:ov-upper-bound}]
We analyze Algorithm~\ref{alg:ov-general}, keeping \(s,t\) as independent parameters until the final bound. Consider a nonempty class \(V_a\). As in the promise case, if a vector with an orthogonal partner has its index in \(V_a\), at least an \(r/n_a\) fraction of the walk states are marked. The spectral gap is \(\Theta(1/r)\). By \Cref{thm:mnrs}, together with \Cref{lem:ov-class-access,lem:ov-general-check}, searching the class and estimating its size costs\footnote{For \(n_a\le2\), prepare the records and run the checker for each walk state, at total cost \(O(S+C)\), within the same bound.}
\begin{align*}
 &\Ot\!\left(n\sqrt{s+n/t}+S+\sqrt{n_a}\,U+\sqrt{n_a/r}\,C\right)\\
 &\hspace{1.5em}=\Ot\!\left(\sqrt{na n_a}
       +n\sqrt{s+n/t}\,(1+n_a^{1/4})\right)\\
 &\hspace{1.5em}\le\Ot\!\left(n\sqrt t+\frac{n^2}{\sqrt{st}}
       +n^{5/4}\sqrt s+\frac{n^{9/4}}t\right).
\end{align*}
Here we use \(r=\Theta(\sqrt{n_a})\), \(a n_a=O(nt+n^3/(st))\), and \(n_a\le n\). The degree sample contributes at most \(n^{7/4}/\sqrt t\), which is at most \(n^{9/4}/t\) since \(t\le n\).

Preprocessing reads the sampled vectors in \(\Ot(ns+n^2/t)\) queries. There are \(O(\log n)\) nonempty size classes, so, absorbing \(n^2/t\) into \(n^{9/4}/t\), the total query cost is
\[
 \Ot\!\left(ns+n\sqrt t+\frac{n^2}{\sqrt{st}}
       +n^{5/4}\sqrt s+\frac{n^{9/4}}t\right).
\]
Choosing \(s=\lceil n^{1/3}\rceil\) and \(t=\lceil n^{5/6}\rceil\) makes this \(\Ot(n^{17/12})\). 
We truncate the algorithm at a sufficiently large constant times
the total query bound above, counting truncation as failure.

Finally, every index belongs to a size class. If an orthogonal pair exists, a state containing either index is marked, and \Cref{lem:ov-general-check} gives a correct checker. If no orthogonal pair exists, no state is marked. The preprocessing and size label events hold with probability at least \(0.97\); reducing the search errors makes the overall success probability at least \(2/3\).
\end{proof}

\section{An existential analogue of verification}
\label{sec:existential-row}

Verification asks whether \emph{every} row of \(A\boolprod B\) equals the
corresponding row of \(C\).  We now ask instead whether \emph{some} row
of \(A\boolprod B\) equals a given vector.

\subsection{Boolean Matrix Product Row Search}
\label{sec:bmprs-bounds}

\begin{definition}[Boolean Matrix Product Row Search]
\label{def:bmprs}
For \(A,B\in\bits^{n\times n}\) and \(v\in\bits^n\),
\begin{equation}
  \BMPRS_n(A,B,v)
  \coloneqq
  \bigvee_{i\in[n]}
  \bigwedge_{j\in[n]}
  \ind{(A\boolprod B)_{ij}=v_j}.
  \label{eq:bmprs-formula}
\end{equation}
All three inputs are accessed through entry queries.  We write
\(\BMPRS_n^{\ones}(A,B)\coloneqq\BMPRS_n(A,B,\ones_n)\) for the
restriction to the all-ones target, whose input consists of \(A\) and
\(B\) only.
\end{definition}

For the all-ones target, verification and row search differ only in
their outermost quantifier:
\begin{align*}
  A\boolprod B=\J_n
  &\quad\Longleftrightarrow\quad
  \forall i\in[n]\ \forall j\in[n]\ \exists t\in[n]:\
  A_{it}=B_{tj}=1,\\
  \BMPRS_n^{\ones}(A,B)=1
  &\quad\Longleftrightarrow\quad
  \exists i\in[n]\ \forall j\in[n]\ \exists t\in[n]:\
  A_{it}=B_{tj}=1.
\end{align*}
The first condition is verification with \(C=\J_n\).  This special
case is as hard as verification up to rescaling (\Cref{thm:bmpv-final}
and the discussion following it), and it can be decided with
\(\widetilde O(n^{17/12})\) queries by
\Cref{thm:bmpv-final,thm:ov-upper-bound}.  In \Cref{sec:bmprs-bounds}, we show
that the second condition requires \(\Omega(n^{3/2}/\sqrt{\log n})\)
queries, while \(O(n^{3/2})\) queries suffice for every target \(v\).
For \(A=B=\widehat A_G\), the two conditions state that \(G\) has
diameter at most two and radius at most two, respectively, by
\Cref{eq:closed-neighborhood-square}.  In \Cref{sec:radius}, we give
reductions in both directions between \(\BMPRS^{\ones}\) and Radius at
Most Two, so the separation carries over to graphs.

\begin{theorem}[Query bounds for Boolean Matrix Product Row Search]
\label{thm:bmprs-bounds}
For every \(n\ge8\),
\[
  \Omega\!\left(\frac{n^{3/2}}{\sqrt{\log n}}\right)
  \le Q(\BMPRS_n^{\ones})
  \le Q(\BMPRS_n)
  =O(n^{3/2}).
\]
\end{theorem}

The upper bound comes from evaluating \Cref{eq:bmprs-formula} level by
level with quantum search.  For fixed \(i\) and \(j\), query \(v_j\) and
use Grover search over \(t\) to decide whether
\((A\boolprod B)_{ij}=\bigvee_{t\in[n]}(A_{it}\booland B_{tj})\) equals
\(v_j\); this costs \(O(\sqrt n)\) queries.  With these bounded-error tests as
inputs, a search over \(j\) for an entry with
\((A\boolprod B)_{ij}\ne v_j\) decides whether row \(i\) of
\(A\boolprod B\) equals \(v\) with \(O(\sqrt n\cdot\sqrt n)\) queries,
and a search over \(i\) decides \(\BMPRS_n\) with
\(O(\sqrt n\cdot\sqrt n\cdot\sqrt n)=O(n^{3/2})\) queries.
The two outer searches use the algorithm of H{\o}yer, Mosca, and de
Wolf~\cite{HoyerMoscaDeWolf2003} for search on bounded-error inputs,
which needs no prior error reduction, so no logarithmic factors arise.

The rest of this subsection proves the lower bound.  As for Orthogonal
Vectors, we start from surjectivity, but the quantifier structure of
\(\BMPRS^{\ones}\) now permits a direct embedding.  A function \(f\) is
not surjective exactly when \(\exists x\,\forall j:f(j)\ne x\), which
matches the two outer quantifiers of \(\BMPRS^{\ones}\).  The innermost
existential quantifier \(\exists t\) asks whether row \(x\) of \(A\) and
column \(j\) of \(B\) share a one.  To make this condition equivalent to
\(f(j)\ne x\), we encode each value \(x\) by a string \(\gamma(x)\) of
length \(O(\log n)\), let row \(x\) of \(A\) mark the zeros of
\(\gamma(x)\), and let column \(j\) of \(B\) hold \(\gamma(f(j))\).

We expand each encoding position into a block of \(\Theta(n/\log n)\)
auxiliary indices, repeating the corresponding column of \(A\)
throughout the block.  For each column of \(B\), the entries in this
block are fresh variables whose OR replaces the original encoded bit.
This implements composition with OR.  Applying adversary composition
multiplies the \(\Omega(n)\) lower bound for surjectivity by
\(\Omega(\sqrt{n/\log n})\).  Unlike the
proof of \Cref{thm:ov-lower-bound}, which places an AND outside many
surjectivity instances, this proof uses a single instance and composes
it with ORs at its inputs.  The length of the encoding accounts for the
factor \(\sqrt{\log n}\) in \Cref{thm:bmprs-bounds}.

\paragraph{Encoding the predicate \(x\ne y\).}
Let \(q\ge4\) be even, and let \(\ell\coloneqq2\ceil{\log_2q}\).  For
\(x\in[q]\), let \(c(x)\in\bits^{\ell/2}\) be the binary representation
of \(x-1\), and define
\[
  \gamma(x)\coloneqq\bigl(c(x),\overline{c(x)}\bigr)\in\bits^\ell.
\]
If \(x\ne y\), then some \(b\in[\ell]\) satisfies \(\gamma(x)_b=0\) and
\(\gamma(y)_b=1\): take a bit on which \(c(x)\) and \(c(y)\) differ, in
the first half of the encoding if \(c(x)\) has a zero there and in the
complemented second half otherwise.  If \(x=y\), no such \(b\) exists.
Hence, for all \(x,y\in[q]\),
\begin{equation}
  \bigvee_{\substack{b\in[\ell]\\\gamma(x)_b=0}}\gamma(y)_b
  =\ind{x\ne y}.
  \label{eq:gamma-inequality}
\end{equation}

For \(y\in\bits^{(2q-2)\times\ell}\), define
\[
  F_q(y)
  \coloneqq
  \bigvee_{x\in[q]}
  \bigwedge_{j\in[2q-2]}
  \bigvee_{\substack{b\in[\ell]\\\gamma(x)_b=0}}y_{j,b}.
\]
For \(\kappa\ge1\), let \(\Phi_{q,\kappa}\coloneqq F_q\circ\OR_\kappa\)
be obtained by replacing each variable \(y_{j,b}\) with the OR of
\(\kappa\) fresh variables; that is, for
\(z\in\bits^{(2q-2)\times\ell\times\kappa}\),
\begin{equation}
  \Phi_{q,\kappa}(z)
  =
  \bigvee_{x\in[q]}
  \bigwedge_{j\in[2q-2]}
  \bigvee_{\substack{b\in[\ell]\\\gamma(x)_b=0}}
  \bigvee_{t\in[\kappa]}z_{j,b,t}.
  \label{eq:encoded-nonsurjectivity}
\end{equation}

\begin{lemma}[Encoded non-surjectivity]
\label{lem:encoded-nonsurjectivity}
For every even \(q\ge4\) and every integer \(\kappa\ge1\),
\[
  Q(\Phi_{q,\kappa})=\Omega(q\sqrt\kappa).
\]
\end{lemma}

\begin{proof}
We first show that \(Q(F_q)=\Omega(q)\).  Given an instance
\(f\from[2q-2]\to[q]\) of \(\ONTO_q\), set the \(j\)-th row of \(y\) to
\(y_{j,*}\coloneqq\gamma(f(j))\).  By \Cref{eq:gamma-inequality},
\[
  F_q(y)=1
  \quad\Longleftrightarrow\quad
  \exists x\in[q]\ \forall j\in[2q-2]:f(j)\ne x,
\]
that is, exactly when \(f\) is not surjective.  Each bit of \(y\)
depends on a single value \(f(j)\), so each query to \(y\) can be
simulated with \(O(1)\) queries to \(f\).  Hence
\(Q(\ONTO_q)=O(Q(F_q))\), and the lower bound of Beame and
Machmouchi~\cite[Corollary~6]{BeameMachmouchi2012} gives
\(Q(F_q)=\Omega(q)\).

As in the proof of \Cref{lem:and-of-onto}, the negative-weight adversary
bound characterizes bounded-error quantum query complexity up to a
constant factor and is multiplicative under
composition~\cite[Theorem~1.1 and Lemma~5.2]{LeeEtAl2011}.  Hence
\(\Adv(F_q)=\Omega(q)\), and since \(\Adv(\OR_\kappa)=\sqrt\kappa\),
\[
  Q(\Phi_{q,\kappa})
  =\Omega\!\left(\Adv(F_q)\,\Adv(\OR_\kappa)\right)
  =\Omega(q\sqrt\kappa).
  \qedhere
\]
\end{proof}

\begin{proof}[Proof of \Cref{thm:bmprs-bounds}]
The upper bound was shown above.  The middle inequality holds because
\(\BMPRS_n^{\ones}\) is a restriction of \(\BMPRS_n\).
For the lower bound, set
\[
  q\coloneqq2\floor{n/4},
  \qquad
  \kappa\coloneqq\floor{n/\ell},
\]
with \(\ell=2\ceil{\log_2q}\) as above.  For \(n\ge8\), the integer \(q\)
is even and at least four, and \(\ell\le q\le2q-2\le n\).  Hence
\(\kappa\ge1\) and \(\ell\kappa\le n\).

We embed \(\Phi_{q,\kappa}\) into \(\BMPRS_n^{\ones}\).  Index the rows
of \(A\) by \(x\in[n]\) and the columns of \(B\) by \(j\in[n]\), and
label \(\ell\kappa\) of the auxiliary indices, shared by the columns of
\(A\) and the rows of \(B\), by the pairs
\((b,t)\in[\ell]\times[\kappa]\).  For \(x\in[q]\), \(j\in[n]\), and
\((b,t)\in[\ell]\times[\kappa]\), define
\[
  A_{x,(b,t)}\coloneqq\ind{\gamma(x)_b=0},
  \qquad
  B_{(b,t),j}\coloneqq
  \begin{cases}
    z_{j,b,t}, & j\in[2q-2],\\
    1, & j\notin[2q-2].
  \end{cases}
\]
All other entries of \(A\) and \(B\) are zero.  For \(x\in[q]\) and
\(j\in[2q-2]\),
\[
  (A\boolprod B)_{x,j}
  =
  \bigvee_{\substack{b\in[\ell]\\\gamma(x)_b=0}}
  \bigvee_{t\in[\kappa]}z_{j,b,t}.
\]
For \(x\in[q]\) and \(j\notin[2q-2]\), we have \((A\boolprod B)_{x,j}=1\),
because \(\gamma(x)\) has \(\ell/2\ge1\) zeros, so row \(x\) of \(A\) is
nonzero.  For \(x\notin[q]\), row \(x\) of \(A\), and hence row \(x\) of
\(A\boolprod B\), is zero.  Comparing with
\Cref{eq:encoded-nonsurjectivity} gives
\(\BMPRS_n^{\ones}(A,B)=\Phi_{q,\kappa}(z)\).  Every entry of \(A\) is
fixed, and every entry of \(B\) is fixed or equal to a single bit of
\(z\).  By \Cref{lem:encoded-nonsurjectivity},
\[
  Q(\BMPRS_n^{\ones})=\Omega(q\sqrt\kappa)
  =\Omega\!\left(n\sqrt{\frac{n}{\log n}}\right)
  =\Omega\!\left(\frac{n^{3/2}}{\sqrt{\log n}}\right),
\]
since \(q=\Theta(n)\), \(\ell=\Theta(\log n)\), and
\(\kappa=\Theta(n/\log n)\).
\end{proof}

\subsection{Radius at Most Two}
\label{sec:radius}

By \Cref{eq:closed-neighborhood-square}, a graph \(G\) has radius at
most two exactly when \(\widehat A_G\boolprod\widehat A_G\) has an
all-ones row, which reduces Radius at Most Two to \(\BMPRS^{\ones}\).
For the reverse reduction, we use a layered construction of Abboud,
Vassilevska Williams, and Wang~\cite{AbboudVassilevskaWilliamsWang2016},
stated directly in terms of the matrices \(A\) and \(B\).  Row, middle,
and column indices become three layers of vertices, and auxiliary
vertices ensure that only row vertices can be within distance two of
every vertex.

\begin{proposition}[Boolean Matrix Product Row Search and radius]
\label{prop:bmprs-radius}
For every \(n\ge1\),
\[
  Q(\RAD_{\le2,n})
  \le Q(\BMPRS_n^{\ones})
  \le Q(\RAD_{\le2,\,3n+3}).
\]
\end{proposition}

\begin{proof}
For the first inequality, let \(G\) be an \(n\)-vertex graph, and set
\(A=B=\widehat A_G\).  Entry \((a,b)\) of \(\widehat A_G\) is \(1\) if
\(a=b\) and is the adjacency bit \((A_G)_{ab}\) otherwise, so it costs
at most one query.  By \Cref{eq:closed-neighborhood-square}, row \(a\)
of \(\widehat A_G\boolprod\widehat A_G\) is all ones exactly when every
vertex is within distance two of \(a\).  Hence
\(\RAD_{\le2,n}(G)=\BMPRS_n^{\ones}(\widehat A_G,\widehat A_G)\).

For the second inequality, let \(A,B\in\bits^{n\times n}\) be an
instance of \(\BMPRS_n^{\ones}\).  The graph \(H\) has \(3n+3\) vertices:
a vertex \(r_i\) for each row index \(i\) of \(A\), a vertex \(w_t\) for
each middle index \(t\), a vertex \(c_j\) for each column index \(j\) of
\(B\), and three auxiliary vertices \(\alpha\), \(\beta\), and
\(\alpha'\).  Its edges are
\begin{itemize}
\item \(\{\alpha,r_i\}\) and \(\{\beta,r_i\}\) for every \(i\in[n]\);
\item \(\{\beta,w_t\}\) for every \(t\in[n]\), and the edge
  \(\{\alpha,\alpha'\}\);
\item \(\{r_i,w_t\}\) whenever \(A_{it}=1\), and \(\{w_t,c_j\}\)
  whenever \(B_{tj}=1\).
\end{itemize}
Only the edges in the last item depend on the input, and each is
answered by a single query to \(A\) or \(B\).

We claim that the vertices within distance two of every vertex of \(H\)
are exactly the vertices \(r_i\) for which row \(i\) of \(A\boolprod B\)
is all ones.  A vertex \(r_i\) is adjacent to \(\alpha\) and \(\beta\),
reaches every other \(r_{i'}\) and \(\alpha'\) through \(\alpha\), and
reaches every \(w_t\) through \(\beta\).  It is not adjacent to \(c_j\),
and their common neighbors are the vertices \(w_t\) with
\(A_{it}=B_{tj}=1\), so
\begin{equation}
  \operatorname{dist}_H(r_i,c_j)\le2
  \quad\Longleftrightarrow\quad
  (A\boolprod B)_{ij}=1.
  \label{eq:bmprs-radius-distance}
\end{equation}
It remains to rule out the other vertices.  The vertex \(\alpha\) and
every \(c_j\) are at distance at least three from each other, because
all neighbors of \(c_j\) have the form \(w_t\) and none of them is
adjacent to \(\alpha\).  The vertex \(\alpha'\) and every vertex
\(\beta\) or \(w_t\) are at distance at least three from each other,
because the only neighbor of \(\alpha'\) is \(\alpha\), and \(\alpha\) is
adjacent to neither.  Hence \(H\) has radius at most two exactly when
\(\BMPRS_n^{\ones}(A,B)=1\).
\end{proof}

\begin{corollary}[Query bounds for Radius at Most Two]
\label{cor:radius-bounds}
For every \(n\ge27\),
\[
  \Omega\!\left(\frac{n^{3/2}}{\sqrt{\log n}}\right)
  \le Q(\RAD_{\le2,n})=O(n^{3/2}).
\]
\end{corollary}

\begin{proof}
The upper bound follows from the first inequality in
\Cref{prop:bmprs-radius} and \Cref{thm:bmprs-bounds}.  For the lower
bound, let \(m\coloneqq\floor{(n-3)/3}\), so that \(m\ge8\) and
\(3m+3\le n\).  Apply the second reduction of \Cref{prop:bmprs-radius} to
an instance of \(\BMPRS_m^{\ones}\), and add \(n-(3m+3)\) further vertices
adjacent only to \(\alpha\).  Each new vertex plays the same role as
\(\alpha'\): it is at distance two from every \(r_i\) and at distance
three from \(\beta\).  Hence the answer is unchanged, and
\Cref{thm:bmprs-bounds} gives
\[
  Q(\RAD_{\le2,n})
  \ge Q(\BMPRS_m^{\ones})
  =\Omega\!\left(\frac{m^{3/2}}{\sqrt{\log m}}\right)
  =\Omega\!\left(\frac{n^{3/2}}{\sqrt{\log n}}\right).
  \qedhere
\]
\end{proof}

Thus Boolean Matrix Product Row Search, already with the all-ones target, and
Radius at Most Two have quantum query complexity
\(\widetilde\Theta(n^{3/2})\), so nested search is optimal for them up
to a factor \(\sqrt{\log n}\).  For their universal counterparts, the
condition \(A\boolprod B=\J_n\) and Diameter at Most Two, nested search
is not optimal: by \Cref{thm:ov-upper-bound} and the reductions of
\Cref{sec:bmpv-and-ov}, both can be decided with
\(\widetilde O(n^{17/12})\) queries.  Changing the outermost quantifier from ``every'' to ``some''
therefore makes both problems polynomially harder.  In particular,
Radius at Most Two is polynomially harder than Diameter at Most Two in
the adjacency matrix model.

\phantomsection
\addcontentsline{toc}{section}{Acknowledgements}
\section*{Acknowledgements}

We thank Andrew Childs for many helpful discussions. ASG 
received support from the National Science Foundation (grant 26-17356) 
and the Department of Energy (grant DE-SC0020264 and the Office of 
Science, Office of Advanced Scientific Computing Research, Accelerated 
Research in Quantum Computing program). FLG is supported by JSPS KAKENHI 
grants JP24H00071 and JP25K24674, MEXT Q-LEAP grant JPMXS0120319794, JST 
ASPIRE grant JPMJAP2302, and JST CREST grant JPMJCR24I4. XZ is
supported by NSERC Grant RGPIN-2024-06493.

\paragraph{Use of generative AI.}
OpenAI GPT 6 Astra and Anthropic Opus 5.5 were used for literature 
searches, manuscript drafting and revision, proof development and 
checking, and LaTeX preparation. The authors retain full 
responsibility for reviewing the generated material and for all 
mathematical claims and final text.

In particular, GPT 6 Astra found bottlenecks in authors' 
earlier arguments for the main upper bound and after significant 
back-and-forth, suggested a preprocessing step that the authors were able
to use to obtain an upper bound of $\tilde{O}(n^{29/20})$ for the promise 
problem considered in \Cref{sec:ov-promise}. It was then used to 
improve the exponents and remove the promise. The final algorithm 
was quite complicated, and it was later simplified by the authors 
resulting in the main algorithm (\Cref{alg:ov-general}) presented in \Cref{thm:ov-upper-bound}.

\begingroup
\small
\raggedright
\bibliographystyle{alphaurl}
\AddToHookNext{cmd/thebibliography/after}{%
  \addcontentsline{toc}{section}{References}%
}
\bibliography{references}
\endgroup

\end{document}